\documentclass[journal]{IEEEtran}
 
\usepackage{amsmath,amssymb,amsfonts}
\usepackage[numbers,sort&compress]{natbib}
\usepackage{graphicx}
\usepackage[hidelinks]{hyperref}
\usepackage{amsthm}
\usepackage{siunitx}
\usepackage{tikz}
\usetikzlibrary{positioning, arrows.meta}
\usepackage{caption}
\usepackage{booktabs}
\usepackage{subcaption}
\usepackage{tabularx}
\usepackage{textcomp}
\usepackage{multirow}
\usepackage{mathtools}
\usepackage{enumitem}
\usepackage{float}
\usepackage{algorithm}
\usepackage{algorithmicx}
\usepackage{algpseudocode}
\usepackage{xurl}
 
\newtheorem{theorem}{Theorem}
\newtheorem{remark}{Remark}
\newtheorem{definition}{Definition}

\setlist[itemize]{leftmargin=*,topsep=2pt,itemsep=1pt,parsep=0pt}

\begin{document}
\raggedbottom

\title{Reusing Spare Vehicle Computing Capacity:\\
Is It Viable, Profitable and Sustainable?}

\author{Rosario~Patan\`e,
        Nadjib~Achir,
        Andrea~Araldo,
        and~Lila~Boukhatem
\thanks{R.~Patan\`e and A.~Araldo are with SAMOVAR, T\'el\'ecom SudParis,
Institut Polytechnique de Paris, Palaiseau, France
(e-mail: rosario.patane@telecom-sudparis.eu;
andrea.araldo@telecom-sudparis.eu).}
\thanks{N.~Achir is with INRIA Saclay, Palaiseau, France
(e-mail: nadjib.achir@inria.fr).}
\thanks{L.~Boukhatem is with Université Paris-Saclay, CNRS, Laboratoire Interdisciplinaire des Sciences du
Numérique, 91400, Orsay, France (e-mail: lila.boukhatem@universite-paris-saclay.fr).}}

\maketitle

\begin{abstract}
Vehicular Cloud Computing (VCC) exploits computing hardware already
embedded in vehicles for purposes unrelated to offloading and puts its
idle cycles to work executing end-users' offloaded tasks, avoiding the
deployment of new computation infrastructure. Despite its conceptual
appeal, adoption is hindered by the lack of quantitative evidence that
sharing spare vehicular capacity is \emph{viable}, \emph{profitable} and
\emph{sustainable}.  A management scheme for
task offloading sustains deadline-constrained execution by decoupling
deadline-aware allocation from ex-post incentive design, making the
approach \emph{viable}. A robust coalitional game and a vehicle-selection
mechanism keep the network operator and the vehicle owners
\emph{profitable} even under unfavorable outcomes of the random
environment, e.g., mobility, radio conditions, or vehicles over-stating
their spare capacity. A simulation campaign on realistic urban mobility,
using Simulation of Urban MObility (SUMO) traces of Rome, together with a
CO\textsubscript{2} life-cycle assessment against edge computing
infrastructure,  whose footprint, is attributed pro rata to the share of
its capacity devoted to offloading shows the approach to be
\emph{sustainable}. Vehicles absorb most of the offloaded traffic within
tens of milliseconds; participation yields up to 157~km of monthly
driving range per vehicle; and life-cycle emissions drop by over 99\% when using VCC
compared to that edge infrastructure.
\end{abstract}


\section{Introduction}
\label{sec:intro}

The paradigm of \emph{ubiquitous computing}, where computational
capabilities are embedded in everyday objects and environments, is
increasingly becoming tangible. As demand grows for low-latency and
compute-intensive services, a key challenge arises: satisfying
performance constraints while minimizing the environmental footprint of
digital infrastructures. This dual goal has driven the emergence of
\emph{sustainable computing} strategies aimed at reducing energy
consumption and CO\textsubscript{2} emissions in distributed
communication systems.

Vehicular Cloud Computing~(VCC), whose underutilized resources enable
scalable, responsive, and sustainable task offloading, is especially
attractive in dense urban areas, where their abundance allows leveraging
mobile computing without extra computing and communication
infrastructure. By \emph{spare capacity} we mean the fraction of a vehicle's on-board computing capacity, already installed for driving-related functions, that the vehicle offers for offloading, so that no new hardware is deployed.

Despite this potential, the practical viability of VCC remains unclear in
terms of consumption, profitability and emissions.

This paper treats the spare computing capacity of vehicles as an
infrastructure that already exists, and asks what it costs and what it
returns to use it. Every tool used here is standard. Their combination
makes three things measurable: a queueing model of every executor turns
an over-stated capacity into a deadline miss that the Controller can
observe; an admission rule filters such executors before allocation; a
coalitional game settles payments on realized outcomes. Delay, energy
and money are measured together on the same traces, so incentives can be
stated in units a vehicle owner understands and the environmental claim
can be tested against an edge server. The paper is organized around
three questions, and the contributions answer them in turn:
\begin{itemize}
  \item \textbf{Is sharing spare capacity viable?}
  A two-stage scheme decouples deadline-constrained allocation, solved
  within a $5$\,ms slot on nominal estimates, from revenue sharing,
  settled off the critical path on realized values. The waiting time an
  executor already carries is part of the delay budget, so
  over-commitment is visible at decision time.
  \item \textbf{Is it profitable?}
  Payments are shared through a coalitional game whose core is non-empty
  at every slot, and an admission rule priced on the gap between declared
  and delivered capacity protects the deadline guarantee when vehicles
  over-declare. We report where the rule pays and where it costs money,
  including the regime in which it does not help and an adaptive variant
  that the data reject, and we express revenues in monthly driving range
  as well as in currency.
  \item \textbf{Is it sustainable?}
  On SUMO traces of Rome we measure per-task energy on the vehicular and
  cloud paths and compare the life-cycle emissions of the fleet with
  those of an edge server, attributing the server footprint pro rata to
  the peak share that offloading occupies on it.
\end{itemize}

The remainder of the paper is organized as follows.
Section~\ref{sec:related_work_last} reviews the related work.
Section~\ref{sec:architecture_last} presents the system architecture.
Section~\ref{sec:modelss} defines the delay, energy, and utility models.
Section~\ref{sec:strategy_allocation} and
Section~\ref{sec:revenue_sharing} detail the task allocation and
revenue-sharing mechanisms. Section~\ref{sec:simulation_last} outlines
the simulation setup. Section~\ref{sec: results-last}
and~\ref{sec:emissions_analysis} report performance, energy, incentive
and emissions results. Section~\ref{sec: conclusion-last} concludes the
paper. The implementation and simulation scripts are available in a
publicly available repository~\cite{patane2026github}.

\section{Related Work}
\label{sec:related_work_last}

\subsection{Task Allocation in Vehicular Systems}
\label{sec:rw1}
Early studies such as~\cite{KumarEnergyVM} revealed energy-performance
trade-offs in mobile offloading. Later, location and delay aware task
allocation strategies were proposed~\cite{XiaDelayMin,Ning2022TMC,Hekmati2020TMC}, while Liu et
al.~\cite{chen23} explored mobility-aware multi-hop task distribution.
Hybrid models were also proposed including joint
communication-computation~\cite{Zhang2022} and mean-field
game-theoretic decision making~\cite{cui22}. 

The open issues listed in~\cite{Ahmed2022Survey}, mobility-induced
volatility of the executor set, heterogeneity of on-board resources and
energy constraints, are the ones that a scheme for real vehicular
settings has to address jointly. Our model explicitly
integrates such constraints into a cooperative optimization strategy.

\subsection{Cooperative Offloading under Real-Time Constraints}
\label{sec:rw2}
AI-based approaches, including deep reinforcement learning~\cite{sun19}
and other adaptive optimization techniques~\cite{liu23}, have been
proposed for vehicular offloading due to their flexibility in dynamic
environments. Cooperative V2V schemes~\cite{chen23} and decentralized VCC
solutions~\cite{wan2020efficient} further improve resource utilization.
However, these methods typically rely on iterative learning, model
training, or extensive signaling, making them difficult to reconcile with
hard real-time and deadline-constrained offloading.

Recent learning-based schemes, such as topology-adaptive offloading using
graph reinforcement learning~\cite{Tang2024TopoOffloading}, improve
average latency but do not provide guarantees on deadline satisfaction or
computational boundedness. As a result, their practical deployability in
latency-critical vehicular scenarios remains limited.

In contrast, our approach adopts lightweight, mathematically grounded
allocation mechanisms with bounded complexity, explicitly designed to
operate within millisecond-scale decision windows while ensuring deadline
compliance.

\subsection{Game-Theoretic Cooperation}
\label{sec:rw3}
Coalitional game theory has been adopted to model economic incentives
among vehicles and stakeholders~\cite{ZhouCoalition,Chattopadhyay2022TMC}, often relying on
online coalition evaluation methods. However, their high computational
cost and reliance on in-loop coalition valuation under dynamic conditions
hinder real-time applicability and practical
deployment~\cite{LeonCalvoPlatoon}. A recent
survey~\cite{Pirhosseini2025GT} further highlights that existing
game-theoretic approaches for Edge/Fog/Cloud offloading largely remain
abstract and lack integrated models jointly accounting for delay, energy,
and profitability. Our work addresses these limitations by decoupling
real-time allocation from incentive computation and performing
cooperative payoff evaluation ex post, enabling lightweight,
delay-aware, and deployability-oriented incentive mechanisms.

\subsection{Robustness to Declared Capacity}
\label{sec:rw4}
The works above take the capacity a vehicle declares in its beacons at
face value, so an executor that over-states its resources is allocated
tasks it cannot complete. Distributionally robust optimization over
Wasserstein ambiguity sets~\cite{Esfahani2018Wasserstein,Kuhn2019Tutorial}
provides admission rules in closed form whose margin grows with the
observed deviation between declared and realized quantities; we adopt
this tool in Sec.~\ref{sec:admission} to filter such executors before
allocation.

\subsection{Positioning of This Work}
This work proposes a cooperative offloading framework that jointly
accounts for latency, energy consumption, and monetary cost. Compared to prior works, we explicitly model stakeholder heterogeneity,
i.e., Network Operator~(NO), Vehicle Owners~(VOs) and Cloud
Computing~(CC), incorporate real-time constraints, and validate
performance using reproducible simulations grounded in urban vehicular
mobility and energy models.

\section{Architecture}
\label{sec:architecture_last}

\begin{figure*}[t]
\centering
\begin{subfigure}[b]{0.40\textwidth}
    \centering
    \includegraphics[width=\linewidth]{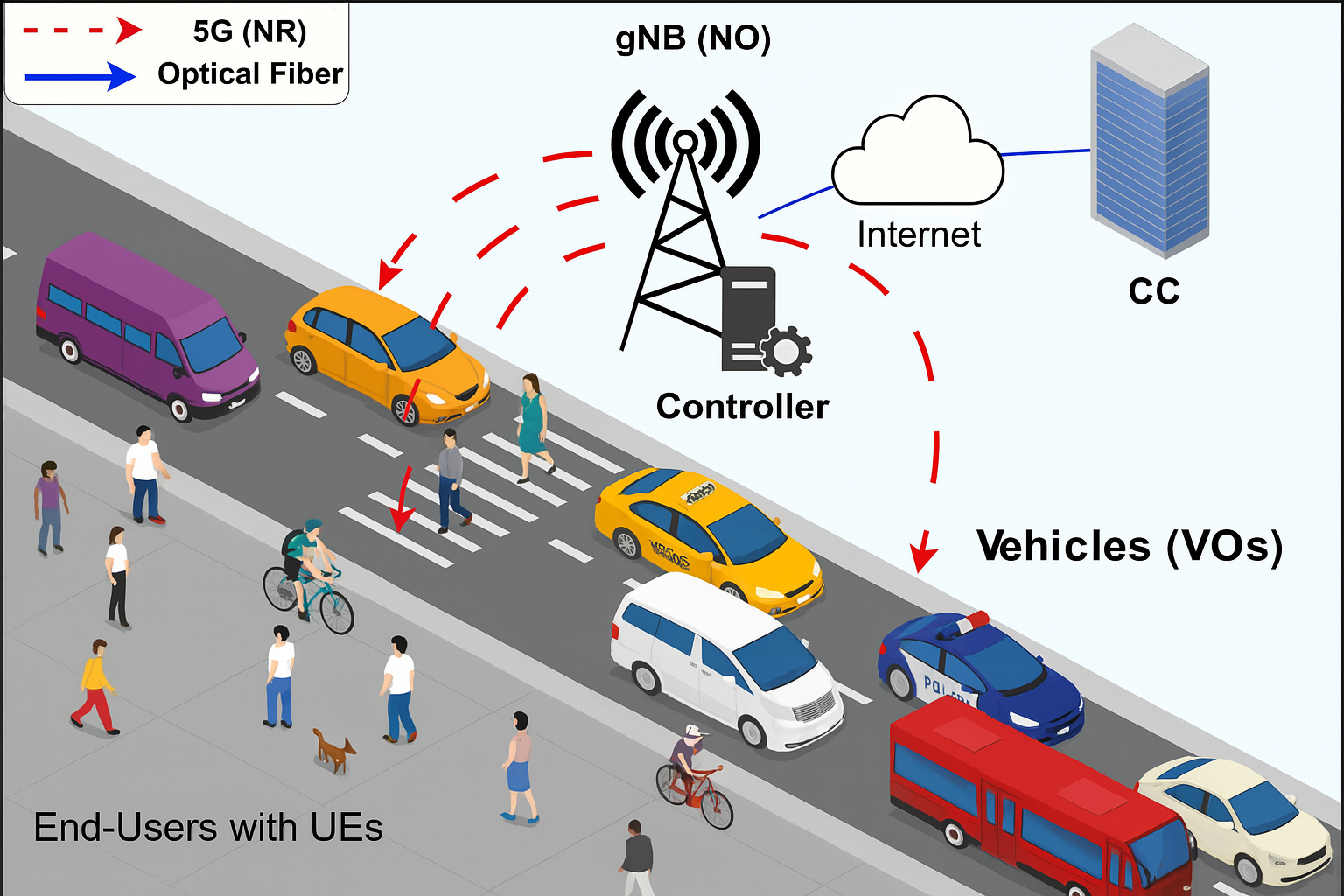}
    \caption{5G-based architecture with centralized control.}
    \label{fig:CoE_MONEY}
\end{subfigure}\hfill
\begin{subfigure}[b]{0.45\textwidth}
\centering
\resizebox{\linewidth}{!}{%
\begin{tikzpicture}[
  stage/.style={rectangle, draw=black!80, fill=blue!10, thick,
                minimum width=3.0cm, minimum height=0.9cm,
                align=center, font=\small},
  process/.style={rectangle, draw=black!60, fill=gray!5,
                  minimum width=2.7cm, minimum height=0.8cm,
                  align=center, font=\footnotesize},
  gamebox/.style={rectangle, draw=black!70, fill=blue!5, thick,
                  minimum width=2.7cm, minimum height=0.8cm,
                  align=center, font=\footnotesize},
  arrow/.style={->, >=stealth, thick}
]
\node[stage] (s1) at (0,0) {\textbf{Stage~1}\\Real-time allocation};
\node[process, below=0.25cm of s1, text width=2.8cm] (rank) {Ranking + robust admission};
\node[process, below=0.25cm of rank, text width=2.8cm] (ilp) {ILP allocation\\(deadline-constrained)};
\node[stage, right=2.6cm of s1] (s2) {\textbf{Stage~2}\\Ex-post sharing};
\node[gamebox, below=0.25cm of s2, text width=2.8cm] (outer) {\textbf{Outer game}\\NO $\tfrac{1}{2}$ $/$ block $\tfrac{1}{2}$};
\node[gamebox, below=0.25cm of outer, text width=2.8cm] (inner) {\textbf{Inner game}\\execs $\propto w_i^\omega$};
\node[process, below=0.25cm of inner, text width=2.8cm] (core) {Core check + LP correction};
\draw[arrow] (s1) -- (rank);
\draw[arrow] (rank) -- (ilp);
\draw[arrow] (s2) -- (outer);
\draw[arrow] (outer) -- (inner) node[midway, right, font=\scriptsize] {$\tfrac{1}{2}\Pi^\omega_\tau$};
\draw[arrow] (inner) -- (core);
\draw[arrow] (ilp.east) -- ++(0.4,0) |- (outer.west)
  node[pos=0.25, right, font=\scriptsize, align=left] {realized\\$T^{i,\omega}_\tau,\,c^i_\tau$};
\draw[arrow, dashed] (s2.north) -- ++(0,0.45) -| (s1.north)
  node[pos=0.25, above, font=\scriptsize] {slot $z\!+\!1$};
\end{tikzpicture}}
\caption{Two-stage scheme with the nested revenue-sharing game.}
\label{fig:two_stage_nested}
\end{subfigure}
\caption{\textbf{(a)} Architecture coordinating vehicular and cloud
resources. \textbf{(b)} Stage~1 allocates on the critical path, Stage~2
settles incentives off it; the stages are pipelined.}
\label{fig:arch_overview}
\end{figure*}

Figure~\ref{fig:CoE_MONEY} illustrates the proposed 5G-based architecture
for vehicular task offloading. Traditional Vehicular Ad-Hoc
Networks~(VANET) face scalability and flexibility issues, limiting
resource efficiency. Centralized solutions using Software-Defined
Networking~(SDN) and Time-Sensitive Networking~(TSN) improve real-time
communication, resource management, and service-oriented control in
vehicular systems~\cite{LUO2025}. For these reasons the proposed
architecture is centralized. The system comprises:
\begin{itemize}
    \item \textbf{End-users:} equipped with mobile 5G User Equipments~(UEs)
    devices issuing offload requests;
    \item \textbf{gNodeB:} 5G base station forwarding tasks to the
    Controller, owned by the \textbf{Network Operator~(NO)};
    \item \textbf{Controller:} co-located with gNodeB, manages beacons,
    task allocation (Sec.~\ref{sec:strategy_allocation}) and revenue
    sharing, owned by the \textbf{NO};
    \item \textbf{Vehicles:} compute-capable 5G nodes broadcasting
    periodic status, owned by the \textbf{Vehicle Owners~(VOs)}, equipped
    with 5G UEs;
    \item \textbf{Cloud:} backup compute node or cluster with consistent
    availability, large computational power, and persistent network
    access, denoted as \textbf{Cloud Computing~(CC)}.
\end{itemize}

\subsection{System Workflow and Interactions}
\label{sec:beacon_task_alloc}

We assume that end-users~(i.e., the UEs) have already decided to offload
to save their battery or to obtain a better Quality of Service~(QoS),
e.g., via mechanisms in~\cite{OffloadTask}.

The \emph{Controller} operates in sequential phases: it begins by
collecting beacons from execution nodes, which include vehicles and, less
frequently (order of minutes), the cloud. Vehicles follow the ETSI CAM
standard~\cite{etsi-cam-spec}, sending beacons at a frequency of
10\,Hz containing their current speed and position. For the same standard
a vehicle becomes unavailable if no beacon is received within
500\,ms~\cite{etsi-cam-spec}.

Vehicles also transmit \emph{aperiodic} beacons upon task completion to
update the controller about their resource availability over time. Each
beacon contains: (i) transmission timestamp, (ii) available computational
resources, (iii) the number of accepted tasks, (iv) processor energy
consumption and (v) radio interface, (vi) cost per kilowatt-hour,
(vii) remaining energy in the battery dedicated for offloading,
(viii) position and speed. Upon receiving these beacons, the
\emph{Controller} aggregates the information and ranks the candidate
nodes based on their \emph{expected dwell time, energy availability}, and
\emph{computational capacity}. The cloud node has infinite dwell time due
to its static nature. Tasks are then assigned to the highest-ranked nodes
that meet the corresponding task deadlines. Once executed, results are
sent from the execution nodes back to the \emph{Controller}, which
forwards them to the end-user.

The system operates in the two stages of Fig.~\ref{fig:two_stage_nested}.
In Stage~1, at every $5$\,ms slot $z$, the Controller ingests the beacons
due in the slot, builds the candidate set as the cloud plus the top-$B$
ranked vehicles, applies the admission rule of Sec.~\ref{sec:admission},
solves the deadline-constrained allocation of
Sec.~\ref{sec:strategy_allocation} on nominal estimates, and dispatches
each task. In Stage~2, once a task completes, the executor reports the
realized completion time and energy, and the Controller recomputes the
per-stakeholder costs on these values, applies the sharing rule of
Sec.~\ref{sec:revenue_sharing} and verifies the core constraints. The
two stages are pipelined: Stage~2 settles the payments of the slots
already served and therefore lies outside the critical path of task
execution, so its running time does not enter the delay budget of any
task.

\section{System Model}
\label{sec:modelss}

\begin{table*}[t]
  \centering\scriptsize\setlength{\tabcolsep}{4pt}\renewcommand{\arraystretch}{1.08}
  \caption{Notation. A hat ($\hat{\cdot}$) marks a nominal quantity, estimated at decision time; a superscript $\omega$ marks its realization.}
  \label{tab:notation}
  \begin{tabular}{@{}l p{6.2cm} l l@{}}
    \toprule
    \textbf{Symbol} & \textbf{Meaning} & \textbf{Unit} & \textbf{Value / where defined} \\
    \midrule
    \multicolumn{4}{@{}l}{\emph{Tasks, slots and sets}}\\
    $z$; $\mathcal{T}_z$ & time-slot; tasks arrived in slot $z$ & --; -- & slot $5$\,ms (Table~\ref{tab:parameters_sim_last}) \\
    $\tau=(I_\tau,O_\tau,W_\tau,D_\tau)$ & task: input bits, output bits, workload, deadline & bit, bit, OP, s & Table~\ref{tab:parameters_sim_last} \\
    $p_\tau=p(D_\tau,W_\tau)$ & end-user payment for task $\tau$ & \$ & $\{2.63,1.43,1.03\}\,\mu\$$ \\
    $\mathcal{V}_z$; $M_{\text{veh}}$ & vehicles in coverage; their number & --; -- & Sec.~\ref{sec:modelss} \\
    $B$ & vehicles admitted to the candidate set & -- & $B=\min(|\mathcal{T}_z|,M_{\text{veh}})$ \\
    $\mathcal{E}_z$; $\mathcal{A}_z$ & candidate executors; executors selected by $\boldsymbol{y}^*_z$ & --; -- & Secs.~\ref{sec:modelss},~\ref{sec:revenue_sharing} \\
    $\mathcal{K}_z$; $\mathcal{N}_z$; $\mathcal{S}$ & stakeholders incl.\ NO; players $\mathcal{A}_z\cup\{\mathrm{NO}\}$; a coalition & --; --; -- & Sec.~\ref{sec:revenue_sharing} \\
    $\mathcal{S}_z$; $\mathcal{T}_z^{\mathcal{S}}$ & players executing $\tau$; tasks served inside $\mathcal{S}$ & --; -- & Eq.~\eqref{eq:allocated_task_set} \\
    $q^{i}$ & per-slot quota of executor $i$ & task & vehicles $1$, cloud unbounded \\
    \midrule
    \multicolumn{4}{@{}l}{\emph{Delay model, Eq.~\eqref{eq:T_unified}}}\\
    $T^{n}_\tau$; $\hat T^{n}_\tau$; $T^{i,\omega}_\tau$ & offloading time on $n$: realized, nominal, at outcome $\omega$ & s & Eq.~\eqref{eq:T_unified} \\
    $\overline{T}$ & mean offloading time over served tasks and seeds & s & Fig.~\ref{fig: rate_off_time} \\
    $R^{\text{5G}}_{\text{UL}},R^{\text{5G}}_{\text{DL}}$; $R^{n}_{\text{UL}},R^{n}_{\text{DL}}$ & UE--gNB rates; gNB--executor rates & bit/s & SINR model, Sec.~\ref{sec:networking} \\
    $d_{\text{UE-NO}}$; $d^{n}$; $d_{\text{cloud}}$ & UE--gNB distance; gNB--executor distance; fiber length to CC & m & $d_{\text{cloud}}=10$\,km \\
    $c_{\text{air}}$; $c_{\text{fib}}$; $c^{n}$ & signal speed in air; in fiber; on the leg to $n$ & m/s & $3\times10^{8}$; $(2/3)c_{\text{air}}$ \\
    $\delta^{n}_{\text{inet}}$; $R_{\text{inet}}$ & one-way core-network latency; backbone rate (cloud only) & s; bit/s & $35$\,ms; $100$\,Gb/s \\
    $C^{n}$ & compute capacity offered by $n$ & OP/s & spare fraction of $3\times10^{14}$ \\
    $T^{n}_{\text{queue}}$ & mean M/G/1 waiting time on $n$ & s & Eq.~\eqref{eq:pk} \\
    $\lambda_n$; $\mu_n$; $\rho_n$ & arrival rate to $n$; service rate; utilization & task/s; task/s; -- & Eq.~\eqref{eq:pk} \\
    $\mathrm{CV}^{2}$ & squared coefficient of variation of the service time & -- & $3.80$ \\
    $T_{\text{alg}}$ & Controller scheduling time, same for every executor & s & measured \\
    \midrule
    \multicolumn{4}{@{}l}{\emph{Energy and cost, Eqs.~\eqref{eq:energy_last_NO}--\eqref{eq: utility_last}}}\\
    $e^{\text{NO}}_\tau$; $e^{n}_\tau$ & energy of the NO; of executor $n$, per task & J & Eqs.~\eqref{eq:energy_last_NO},~\eqref{eq: energy_node_last} \\
    $P^{\text{NO}}_{\text{tx}}$; $P^{\text{NO}}_{\text{CPU}}$; $P^{n}_{\text{tx}}$ & gNB transmit power; Controller CPU power; executor transmit power & W & $23$\,dBm; $200$\,W; $23$\,dBm \\
    $\zeta^{n}$ & energy per operation of $n$ & J/OP & $4.35$ / $4.24\times10^{-13}$ (veh./cloud) \\
        $P_{\text{circ}}$; $\eta_{\text{PA}}$ & EARTH circuit power; power-amplifier efficiency (vehicle radio link) & W; -- & \cite{Auer2011} \\
    $\xi_{\text{inet}}$ & Internet transmission intensity per bit (cloud link) & J/bit & $2.7\times10^{-6}$~\cite{Aslan2018} \\\\
    $\chi^{i}$ & energy-to-cash rate of stakeholder $i$ & \$/J & $0.21$\,\$/kWh \\
    $c^{i}_\tau$; $c^{\text{off},i}_\tau$ & cost of $i$ for $\tau$; executor-plus-NO cost & \$ & Eq.~\eqref{eq: utility_last} \\
    $\boldsymbol{\theta}$; $\hat{\boldsymbol{\theta}}$; $\omega\in\Omega$ & random parameter vector; its nominal value; an outcome & --; --; -- & Sec.~\ref{sec:energy_models} \\
    \midrule
    \multicolumn{4}{@{}l}{\emph{Allocation and admission, Secs.~\ref{sec:strategy_allocation}--\ref{sec:admission}}}\\
    $y^{i}_{\tau,z}$; $\boldsymbol{y}^{*}_z$ & assignment of $\tau$ to $i$; optimal allocation & $\{0,1\}$ & Eq.~\eqref{eq: optimization_1} \\
    $s^{i}_\tau$ & normalized deadline slack of the pair & -- & Eq.~\eqref{eq:slack} \\
    $\Gamma$ & objective scaling (argmax-invariant) & -- & $10^{6}$ \\
    $\mathbf{1}_{i\tau}$; $\mathbf{1}^{\omega}_{i\tau}$ & deadline-success indicator: nominal; realized & $\{0,1\}$ & Sec.~\ref{sec:revenue_sharing} \\
    $C^{i}_{\text{decl}}$; $\bar C^{i}_{\text{real}}$ & declared capacity; mean realized capacity over history & OP/s & beacons \\
    $d^{i}_W$ & relative over-declaration of $i$ (Wasserstein gap) & -- & Eq.~\eqref{eq:dW} \\
    $\lambda_W$ & price of one unit of $d^{i}_W$, as a fraction of $p_\tau$ & -- & $5.0$ \\
    $\alpha$; $\kappa(\alpha)$ & CVaR level; Gaussian CVaR coefficient $\phi(\Phi^{-1}(\alpha))/(1-\alpha)$ & --; -- & $0.90$; $1.7549$ \\
    $\phi$; $\Phi$ & standard Gaussian density; distribution function & --; -- & -- \\
    $\epsilon$ & Wasserstein radius on the cost law & \$ & $0$ (misreporting isolated) \\
    $\hat c^{\text{off},i}_\tau$; $\hat\sigma^{\text{off},i}_\tau$ & nominal mean and st.\ dev.\ of the offloading cost & \$ & $\hat\sigma=0.1\,\hat c$ \\
    \midrule
    \multicolumn{4}{@{}l}{\emph{Revenue sharing, Sec.~\ref{sec:revenue_sharing}}}\\
    $v_z$; $v^{\omega}_z(\mathcal{S})$ & coalition value function; its realization at $\omega$ & \$ & Eqs.~\eqref{eq:value_function_NO_alone}--\eqref{eq:value_function_NO_in} \\
    $\boldsymbol{x}=(x^{i})$ & payoff allocation (imputation when in the core) & \$ & Sec.~\ref{sec:revenue_sharing} \\
    $\beta^{i}$ & share of $p_\tau$ assigned to player $i$, $\sum_i\beta^i=1$ & -- & induced by the sharing rule \\
    $\Pi^{\omega}_\tau$ & realized payoff of task $\tau$ & \$ & $p_\tau-\sum_i c^{i}_\tau(\boldsymbol{\theta}^\omega)$ \\
    $\pi^{\omega}_{\mathrm{NO}}(\tau)$; $\pi^{\omega}_i(\tau)$ & payoff of the NO; of executor $i$ & \$ & Eqs.~\eqref{eq:pi_NO},~\eqref{eq:pi_i_cost} \\
    $w^{\omega}_i$ & cost-proportional weight of executor $i$ & -- & Eq.~\eqref{eq:pi_i_cost} \\
    $G^{\omega}_{z,\boldsymbol\beta}$; $\mathcal{C}$; $\gamma_{\mathcal{S}}$ & TU game; set of coalitions; balanced-collection weights & --; --; -- & Theorem~\ref{thm:core_nonempty_BS_beta} \\
    \midrule
    \multicolumn{4}{@{}l}{\emph{Evaluation, Secs.~\ref{sec:simulation_last}--\ref{sec:emissions_analysis}}}\\
    $\psi$ & fraction of vehicles over-declaring their spare capacity & -- & $\{0,0.2,0.4,0.6,0.8\}$, intensity $0.6$ \\
    $f_{\text{peak}}$ & peak compute share of offloading on the edge server & -- & measured, Sec.~\ref{sec:peak} \\
    \bottomrule
  \end{tabular}
\end{table*}

In this section, we formalize the \emph{delay} and \emph{energy
consumption} models underlying the task allocation process;
Table~\ref{tab:notation} collects the notation used throughout the
paper. Then, the
interactions among the stakeholders~(Network Operator, Cloud Computing,
Vehicle Owners)$=$(NO, CC, VOs) are formalized through a cost model.
At each time-slot~$z$, if tasks~$\mathcal{T}_z$ arrived from end-users,
the NO considers set~$\mathcal{E}_z$ of nodes as potential executors. The
cloud is always in~$\mathcal{E}_z$. Then the~$B$ vehicles highest in the
ranking (see Sec.~\ref{sec:beacon_task_alloc}) are inserted
in~$\mathcal{E}_z$, where~$B=\min(|\mathcal{T}_z|,M_\text{veh})$
and~$M_\text{veh}$ is the number of vehicles currently in coverage and
available for offloading.

\subsection{Offloading Time}
\label{sec:time_models}

A task is modeled as a tuple $\tau = (I_\tau, O_\tau, W_\tau, D_\tau)$,
where $I_\tau$ and $O_\tau$ denote the input and output, including data
and code, measured in~[bits]. The term $W_\tau$ is the computational
workload measured in operations~[OP], and $D_\tau$ is the execution
deadline, in [seconds]. The set of tasks to be offloaded at time-slot $z$
is denoted as $\mathcal{T}_z$. A task can be executed by an executor node
$n$, which can be either the CC or a vehicle $v \in \mathcal{V}_z$, where
$\mathcal{V}_z$ is the set of offloading-capable vehicles at time slot
$z$. All parameter values used in the evaluation are reported and
justified in Table~\ref{tab:parameters_sim_last}. The offloading time for
task $\tau$ executed by $n$ is:

\begin{align}
T^{n}_{\!\tau} =\; &
\underbrace{\frac{I_\tau}{R_{\text{UL}}^{\text{5G}}}
+ \frac{O_\tau}{R_{\text{DL}}^{\text{5G}}}}_{\text{UE-NO transmission}}
+ \underbrace{2\,\frac{d_{\text{UE-NO}}}{c_{\text{air}}}}_{\text{Radio propagation}}
\notag\\[4pt]
&+ \underbrace{\frac{I_\tau}{R_{\text{UL}}^n}
+ \frac{O_\tau}{R_{\text{DL}}^n}}_{\text{NO-}n\text{ transmission}}
+ \underbrace{2\,\frac{d^{n}}{c^n}
+ 2\,\delta_{\text{inet}}^{n}}_{\text{Network delay}}
\notag\\[4pt]
&+ \underbrace{\frac{W_\tau}{C^{n}}}_{\text{Computation}}
+ \underbrace{T^{n}_{\text{queue}}}_{\text{Queueing}}
+ \underbrace{T_{\text{alg}}}_{\text{Allocation decision}}.
\label{eq:T_unified}
\end{align}

The offloading time $T^{n}_{\tau}$ in~\eqref{eq:T_unified} includes the
end-user to NO (UE-NO) transmission time, given by the uplink and
downlink 5G rates $R_{\text{UL}}^{\text{5G}}$ and
$R_{\text{DL}}^{\text{5G}}$, followed by the radio propagation delay over
distance $d_{\text{UE-NO}}$ at the speed of light in air
$c_{\text{air}}$. The data transfer between NO and executor node $n$
(vehicle or cloud) is characterized by rates $R_{\text{UL}}^n$,
$R_{\text{DL}}^n$, and a propagation delay over distance $d^n$ at speed
$c^n$. If $n$ is a vehicle, $c^n = c_{\text{air}}$; if $n=\text{CC}$,
$c^n= c_{\text{fib}}$ denotes the speed of light in optical fiber. The
term $2\,\delta_{\text{inet}}^{n}$ accounts for the core-network latency
of the Internet leg, which is null for vehicles as no Internet routing
occurs. The computation time is $W_\tau / C^n$, $T^n_{\text{queue}}$ is
the waiting time before service, and $T_{\text{alg}}$ the measured
Controller scheduling time, identical for every executor.

An executor with pending tasks cannot serve $\tau$ immediately. We model
every executor as a single-server queue and take $T^{n}_{\text{queue}}$
as the Pollaczek--Khinchine mean waiting time of the M/G/1
queue~\cite{Kleinrock1975,Miao2023TMC},
\begin{equation}
  T^{n}_{\text{queue}}
  =
  \frac{1+\mathrm{CV}^{2}}{2}\,\frac{\rho_n}{\mu_n-\lambda_n},
  \quad
  \mu_n=\frac{C^{n}}{\mathbb{E}[W_\tau]},
  \quad
  \rho_n=\frac{\lambda_n}{\mu_n},
  \label{eq:pk}
\end{equation}
where $\mu_n$ is the service rate, $\lambda_n$ the arrival rate routed to
$n$, tracked online by an exponentially weighted moving average, $\rho_n$
the utilization, and $\mathrm{CV}^{2}\approx 3.8$ the squared coefficient
of variation of the service time, inherited from the workload law of
Table~\ref{tab:parameters_sim_last}. An idle node contributes
$T^{n}_{\text{queue}}=0$; a saturated one contributes an unbounded wait,
so its tasks miss the deadline. Because $\mu_n$ depends on the declared
$C^{n}$, the queue turns a capacity misreport into an observable deadline
violation, which Sec.~\ref{sec:robustness-results} quantifies.

The formula is exact in expectation and, by PASTA~\cite{Wolff1982},
Poisson arrivals sample the time averages; deviations above the mean are
filtered ex ante by the admission rule of Sec.~\ref{sec:admission} and
penalized ex post by the deadline indicator of
Sec.~\ref{sec:revenue_sharing}. A worst-case bound is not used because
the M/G/1 waiting time is \textbf{unbounded}, hence uninformative for admission.

\subsection{Energy and Cost Model}
\label{sec:energy_models}

For energy modeling, we exclude the UE uplink, since user equipment
energy consumption is exogenous to the system and cannot be controlled or
optimized by it. Consequently, the analysis focuses on infrastructure-side
energy consumption. We define energy consumption per task $\tau$ for each
architecture stakeholder as follows.

\paragraph{Network Operator}
The NO's energy consumption includes transmission to executor node $n$
(via 5G wireless link if $n$ is a vehicle, or through the Internet if
CC), execution of the task allocation algorithm, and 5G downlink
transmission from the gNB to the UE.

\begin{equation}
\label{eq:energy_last_NO}
e^{\text{NO}}_\tau =
\underbrace{
P_{\text{tx}}^{\text{NO}} \left(
\dfrac{I_\tau}{R_{\text{UL}}^n}
+
\dfrac{O_\tau}{R_{\text{DL}}^{\text{5G}}}
\right)
}_{\text{Transmission}}
+
\underbrace{
P^{\text{NO}}_{\text{CPU}}\,T_{\text{alg}}
}_{\text{Decision}} \quad [J]
\end{equation}

where $P_{\text{tx}}^{\text{NO}}$ is the NO's gNB transmit power and
$P^{\text{NO}}_{\text{CPU}}$ is the \emph{Controller} CPU power draw.

\paragraph{Executors}
Vehicles and CC consume energy for \emph{task computation} and
\emph{output transmission}:

\begin{equation}
\label{eq: energy_node_last}
e^{n}_\tau =
\underbrace{\zeta^{n}\,W_\tau}_{\text{Computation}}
\;+\;
\underbrace{
\begin{cases}
\left(P_{\text{circ}} + \dfrac{P_{\text{tx}}}{\eta_{\text{PA}}}\right)
\dfrac{O_\tau}{R^{n}_{\text{UL}}}, & n \in \mathcal{V}_z \ \text{(vehicle)} \\[8pt]
\xi_{\text{inet}}\, O_\tau, & n = \text{CC (cloud)}
\end{cases}}_{\text{Transmission}}  \quad [J]
\end{equation}

Compute energy is metered per operation: accelerators are rated in operations per joule and a short
inference does not draw the full thermal design power for its whole
duration, so $\zeta^{n}$~[J/OP] is the inverse of the accelerator
efficiency of node $n$. The values of Table~\ref{tab:parameters_sim_last}
are taken from the datasheet efficiency of the NVIDIA Orin class of
automotive accelerators for vehicles and of the data-center accelerator
adopted for the cloud node~\cite{nvidia_orin}. Transmission energy is
metered differently on the two paths, since they are physically
different links. A vehicle transmits over the radio interface, so the
EARTH model~\cite{Auer2011} applies, with $P_{\text{circ}}$ the circuit
power, $P_{\text{tx}}$ the radiated power and $\eta_{\text{PA}}$ the
power-amplifier efficiency, parameterized as in~\cite{Auer2011}. The
cloud node instead reaches the Controller over a wired backhaul, where
neither a power amplifier nor a radio rate applies; its transmission
energy is charged at a constant intensity $\xi_{\text{inet}}$~[J/bit] per
transported bit, measured for Internet data
transfer~\cite{Aslan2018}, independent of $R^{n}_{\text{UL}}$. Beacon energy consumption is not
considered, since beacons are embedded in CAM messages already
transmitted for driving-related applications, independently of
offloading.

Energy consumption is influenced by several random variables, including
\textit{channel quality}, \textit{transmission power}, and \textit{task
execution time}. These variables are correlated with vehicular mobility.
For instance, the computational capacity of a vehicle depends on its
current workload: if it is already performing internal computation (for
instance related to driving), fewer computing resources remain available
for newly offloaded tasks. Similarly, the mobility pattern (position and
speed) affects network conditions, and coverage may fluctuate due to
changes in speed or position within the mobility scenario. In the case of
CC, the randomness primarily stems from the network condition and the
backhaul traffic. The \textbf{offloading time} and the corresponding
\textbf{energy consumption} are therefore random variables. We formalize
these uncertainties through a parameter vector defined, per each
time-slot~$z$, as
\(\boldsymbol{\theta} = (T^{n}_{\!\tau}, e^{n}_\tau,
e^{\text{NO}}_\tau)_{n\in\mathcal K_z, \tau\in\mathcal T_z}\),
where each component represents a random variable. The term
$\mathcal K_z$ represents all the stakeholders including the NO. In
contrast, the task parameters
\(\tau = (I_\tau, O_\tau, W_\tau, D_\tau)\) are assumed to be fixed upon
arrival.

\textbf{Cost Model.} For each task \(\tau \in \mathcal{T}_z\) and slot
\(z \in Z\), the energy-related cost suffered by each stakeholder
\( i \in \mathcal{K}_z = \{\text{NO}, \text{CC}, \text{VO}^1, \dots,
\text{VO}^m\}_z \) is a random variable:

\begin{equation}
\label{eq: utility_last}
c^i_\tau(\boldsymbol{\theta}) = \chi^{i}\, e^i_\tau(\boldsymbol{\theta})
\quad [\$]
\end{equation}

where \( \chi^{i} \) denotes the energy-to-cash conversion rate [\$/J],
and \( e^i_\tau(\boldsymbol{\theta}) \) is the stochastic energy
consumption of \( i \). In the following,
$\boldsymbol \theta^\omega$ denotes the realization
of~$\boldsymbol \theta$, for outcome~$\omega\in\Omega$.

\section{Task Allocation Problem}
\label{sec:strategy_allocation}

We define the system-wide utility \( f_z \) as the total value generated
at time-slot \( z \in Z \), through the assignment of tasks
\( \tau \in \mathcal{T}_z \) to the executor nodes
\( \mathcal{E}_z = \mathcal{K}_z \setminus \{\text{NO}\} \). Although the
NO does not execute tasks, it participates in every offloading operation
via the \textit{Controller} and the gNodeB (orchestration and
transmission), incurring a per-task cost
\( c^\text{NO}_{\tau,z}(\boldsymbol{\theta})\).

Each task \(\tau\) yields payment \(p(D_\tau, W_\tau)\) from the
end-user. This payment is received by the NO and will later be
redistributed according to the approach described in
Sec.~\ref{sec:stochasticity_last}. Executor nodes do not receive direct
payments but are indirectly compensated through the cooperative
cost-sharing scheme. Let
\(c^{\text{off,}i}_{\tau,z}(\omega) = c^i_{\tau,z}(\omega) +
c^\text{NO}_{\tau,z}(\omega)\)
denote the total cost incurred by executor \(i\) and the NO for the
offloading of task \(\tau\) at time-slot \(z\), under realization
\(\omega\).

The NO needs to decide allocation
$\boldsymbol{y}_z=(y_{\tau,z}^i)_{\tau\in\mathcal{T}_z,i\in\mathcal{E}_z}$,
at time-slot~$z$, where~$y_{\tau,z}^i$ is~$1$ if the task is allocated to
executor~$i$, and~$0$ otherwise. The realizations of the random variables
are not known at the moment of the decision, which has thus to be taken
based on nominal parameter values~\(\hat{\boldsymbol{\theta}}\). Some
parameters are estimated directly by the NO, e.g., network conditions
such as~$R_\text{UL}^\text{5G}$. Other parameters, such as time or energy
($T_\tau^i,e_\tau^i$), are estimated based on information reported by the
executor nodes, e.g., position, transmit power~$P_\text{tx}^i$, and
capacity available for offloading~$C^i$, via
Eqs.~\eqref{eq:T_unified}--\eqref{eq: energy_node_last}.

The allocation ranks candidate pairs by their expected margin weighted by
the residual slack they leave before the deadline. Let

\begin{equation}
\label{eq:slack}
s^{i}_{\tau} =
\frac{D_\tau - \hat T^{i}_{\tau}}{D_\tau} \in (0,1]
\end{equation}

be the normalized slack of pair $(\tau,i)$, computed from the estimated
completion time of Eq.~\eqref{eq:T_unified}, which already includes the
queueing term~\eqref{eq:pk}. Decision~$\boldsymbol{y}_z^*$ is the result
of the following optimization:

\begin{equation}
\label{eq: optimization_1}
\boldsymbol{y}_z^*  = \underset{\boldsymbol{y}_z}{\text{argmax}}
\sum_{i \in \mathcal{E}_z} \sum_{\tau \in \mathcal{T}_z}
\left( p(D_\tau, W_\tau) -  c^{\text{off},i}_{\tau,z}(\hat{\boldsymbol{\theta}}) \right)
s^{i}_{\tau}\; y^i_{\tau,z}
\end{equation}

subject to:
\begin{align}
& \sum_{i \in \mathcal{E}_z} y^i_{\tau,z} \leq 1,
&& \forall \tau \in \mathcal{T}_z
&& \text{\small (one executor per task)}
\label{eq:one_executor_per_task} \\
& \sum_{\tau \in \mathcal{T}_z} y^i_{\tau,z} \leq q^i,
&& \forall i \in \mathcal{E}_z
&& \text{\small (executor capacity)}
\label{eq:executor_capacity}
\end{align}

where $q^i$ is the per-slot quota, declared in the beacon, of tasks that $i$ accepts. Pairs whose
estimated completion time exceeds the deadline are excluded from the
candidate set before the optimization. The slack weight $s^{i}_{\tau}$
encodes a cost that energy alone misses: consuming the whole time budget
exposes the task to a late failure and congests the executor for those
that follow, so between two equally profitable pairs the optimizer
prefers the one leaving more margin. Being a constant of the pair,
computed before the optimization, it keeps
problem~\eqref{eq: optimization_1}--\eqref{eq:executor_capacity} a linear
integer program, and the guarantees of Sec.~\ref{sec:revenue_sharing} are
unaffected by the weighting. 

After offloading and task execution, the executor nodes report the actual
parameter values via updated beacons. The NO can thus recompute consumed
energy and actual execution times, which may differ from the nominal ones
used at the moment of decision. Revenue sharing is based on such updated
values (Sec.~\ref{sec:revenue_sharing}). Executors that systematically
over-declare their spare capacity are filtered before allocation by the
admission rule of Sec.~\ref{sec:admission}.

\subsection{Robust Admission}
\label{sec:admission}

Beacons may over-state the capacity a vehicle actually devotes to
offloading. An inflated capacity shortens the declared waiting time
of~\eqref{eq:pk} while lengthening the realized one, so the discrepancy
surfaces as a deadline miss only after the task has been dispatched. The
admission test guards against it before the pair enters the optimization,
using two quantities. The first is the over-declaration of executor $i$,
the 1-Wasserstein distance between the point mass at the declared
capacity $C^{i}_{\text{decl}}$ and the point mass at the mean realized
capacity $\bar C^{i}_{\text{real}}$ over its observed history, normalized
by the latter and clipped to its adversarial side:
\begin{equation}
\label{eq:dW}
      d_W^{i} \;=\; \max\!\Bigl(0,\;\bigl(C^{i}_{\text{decl}}-\bar C^{i}_{\text{real}}\bigr)/\bar C^{i}_{\text{real}}\Bigr),
\end{equation}
zero for a truthful or conservative executor and growing with the
fraction by which capacity is overstated. The second is the tail of the
offloading cost: with $\hat c^{\text{off},i}_\tau$ and
$\hat\sigma^{\text{off},i}_\tau=0.1\,\hat c^{\text{off},i}_\tau$ the
nominal mean and standard deviation of $c^{\text{off},i}_{\tau,z}$, the
Conditional Value-at-Risk at level $\alpha$ of a Gaussian cost is
$\hat c^{\text{off},i}_\tau+\kappa(\alpha)\,\hat\sigma^{\text{off},i}_\tau$
with $\kappa(\alpha)=\phi(\Phi^{-1}(\alpha))/(1-\alpha)$, $\phi$ and
$\Phi$ the standard Gaussian density and distribution
function~\cite{RockafellarUryasev2000}.

\begin{definition}[Robust admission rule]
\label{def:admission}
Task $\tau$ is admitted to executor $i$ at slot $z$ only if
\begin{equation}
\label{eq:adm_closed}
\underbrace{\hat c^{\text{off},i}_\tau}_{\text{mean cost}} + \underbrace{\kappa(\alpha)\,\hat\sigma^{\text{off},i}_\tau}_{\text{tail margin}} + \underbrace{\tfrac{\epsilon}{1-\alpha}}_{\text{ambiguity}} + \underbrace{\lambda_W p_\tau d_W^{i}}_{\text{misreport.\ penalty}} \leq \underbrace{p_\tau}_{\text{payment}}.
\end{equation}
\end{definition}

Both $d_W^{i}$ and $\lambda_W$ are dimensionless, the penalty taking its
monetary scale from $p_\tau$: the rule is invariant under a rescaling of
the tariff, and $\lambda_W$ reads as the fraction of the payment forfeited
per unit of relative over-declaration. The radius $\epsilon$, declared on the law of
the scalar cost and carrying the units of $p_\tau$, hedges a generic
mismatch between the nominal cost distribution and the true one, and
enters through the margin $\epsilon/(1-\alpha)$, an instance of a known
regularization bound for worst-case expectations over Wasserstein
balls~\cite[Th.~6.3 and Prop.~6.5]{Esfahani2018Wasserstein}, the factor
$1/(1-\alpha)$ being the Lipschitz modulus of the CVaR integrand and, the
cost being scalar and unconstrained in sign, the bound holding with
equality; the derivation and this tightness argument are in the
supplementary material. The term $\lambda_W p_\tau d_W^{i}$
instead prices a specific, executor-attributable discrepancy, itself a
Wasserstein distance~\eqref{eq:dW}, between what a node declares and what
it has delivered. Throughout the evaluation we set $\epsilon=0$, so any protection measured in Sec.~\ref{sec:robustness-results} is attributable to the misreporting term alone.

Algorithm~\ref{alg:stage1}
summarizes the per-slot procedure; the admission loop costs
$O(|\mathcal{T}_z|\,|\mathcal{E}_z|)$ scalar tests and the only
optimization is the final integer program over the admitted pairs, which
at the loads of Table~\ref{tab:parameters_sim_last} carries a few tens of
binaries and runs well under the $5$\,ms slot.

\begin{algorithm}[t]
\caption{Stage~1: real-time allocation at slot $z$.}
\label{alg:stage1}
\begin{algorithmic}[1]
\Require tasks $\mathcal{T}_z$; beacon history; parameters
  $(\alpha,\epsilon,\lambda_W)$; payments $p_\tau$
\Ensure allocation $\boldsymbol{y}_z^{*}$
\State collect the beacons of the slot and update the executor ranking
\State $\mathcal{E}_z \gets \{\mathrm{CC}\}\cup\text{top-}B\text{ ranked vehicles}$,
  with $B=\min(|\mathcal{T}_z|,M_{\text{veh}})$
\ForAll{pairs $(\tau,i)\in\mathcal{T}_z\times\mathcal{E}_z$}
  \State estimate $\hat T^{i}_\tau$ by Eq.~\eqref{eq:T_unified} and
    $\hat c^{\text{off},i}_\tau$ by
    Eqs.~\eqref{eq:energy_last_NO}--\eqref{eq: utility_last}
  \State compute the over-declaration $d_W^{i}$ of Eq.~\eqref{eq:dW}
    from the beacon history
  \If{$\hat T^{i}_\tau\le D_\tau$ \textbf{and} the admission test of
    Eq.~\eqref{eq:adm_closed} holds}
    \State mark $(\tau,i)$ as admissible
  \EndIf
\EndFor
\State $\boldsymbol{y}_z^{*}\gets$ solve the
  ILP~\eqref{eq: optimization_1}--\eqref{eq:executor_capacity} over the
  admissible pairs
\State dispatch each task to its executor; \Return $\boldsymbol{y}_z^{*}$
\end{algorithmic}
\end{algorithm}

\section{Revenue Sharing via Cooperative Game Theory}
\label{sec:revenue_sharing}

At every time-slot \( z \in Z \), tasks are allocated through
optimization~\eqref{eq: optimization_1}--\eqref{eq:executor_capacity},
and the resulting utility must be redistributed among the involved
stakeholders. We model revenue sharing via coalitional game theory. Let
\begin{equation}
\mathcal{A}_z = \left\{ \textstyle i \in \mathcal{E}_z \mid
\sum_{\tau \in \mathcal{T}_z} y^{i,*}_{\tau,z} \geq 1 \right\}
\end{equation}

be the subset of executors that are effectively selected by the optimal
allocation \(\boldsymbol{y}_z^*\) at time-slot \(z\). The set of players
is \( \mathcal{N}_z = \mathcal{A}_z \cup \{\text{NO}\} \). A classical
coalitional game is defined by the pair \( (\mathcal{N}_z, v_z) \), where
\( v_z : 2^{\mathcal{N}_z} \to \mathbb{R} \) assigns a value to each
possible coalition \( \mathcal{S} \subseteq \mathcal{N}_z \),
representing the maximum utility that \( \mathcal{S} \) can achieve
independently of the remaining players. The grand coalition corresponds
to the full set \( \mathcal{N}_z \), and the domain
\( 2^{\mathcal{N}_z} \) is the set of all its \( 2^{|\mathcal{N}_z|} \)
possible subgroups. A payoff allocation
\( x \in \mathbb{R}^{\mathcal{N}_z} \) represents the revenue assigned to
each player. It is \textbf{efficient} if it fully distributes the value
of the grand coalition, i.e.,
\( \sum_{i \in \mathcal{N}_z} x^i = v_z(\mathcal{N}_z) \), and
\textbf{coalitionally rational} if no group of players receives less than
its own value, i.e.,
\( \sum_{i \in \mathcal{S}} x^i \ge v_z(\mathcal{S}) \) for all
\( \mathcal{S} \subseteq \mathcal{N}_z \). Any allocation satisfying both
properties belongs to the \emph{core} of the game, the set of stable
distributions under which no subset of players has an incentive to
deviate from the grand coalition.

\begin{remark}[End-users reward]
End-users are considered through strict deadline constraints: tasks are
accepted and executed only if their end-to-end delay requirements are
met, with latency minimization as their implicit reward, ensuring QoS.
\end{remark}

\subsection{Robust Coalitional Game}
\label{sec:stochasticity_last}

The NO plays a central role in payment management since the payments pass
through it. As both the infrastructure owner and the communication
coordinator, the NO is responsible for solving the task allocation
problem~\eqref{eq: optimization_1}. The value of a coalition
\( \mathcal S \subseteq \mathcal N_z \) is defined via a \emph{value
function} that evaluates the realized utility for each coalition after a
specific outcome \( \omega \in \Omega \), once the task allocation
decision has been made on nominal
values~(Sec.~\ref{sec:strategy_allocation}). The coalition value function
\( v_z^\omega(\cdot) \) is stochastic, due to the deviation from nominal
values that impacts the costs and the collected payments. We formalize a
game with \textit{Transferable Utility}~(TU), since the value collected by the coalition
can be \emph{transferred} among players. Given a weight
$\beta^i \in [0,1]$ multiplying $p(D_\tau, W_\tau)$, we represent the
proportion of payment given to the NO or to the other players, with
$\sum_{i\in\mathcal{N}_z} \beta^i =1$.

In the evaluation, we do not fix the weight vector $\boldsymbol{\beta}=(\beta^i)_{i\in\mathcal N_z}$
directly, but define a payoff rule based on the realized outcome. Let
$p_\tau$ denote the payment associated with task
$\tau \in \mathcal{T}_z$ and let $c^{i}_\tau(\boldsymbol{\theta}^\omega)$
be the realized cost of player $i \in \mathcal{N}_z$ under outcome
$\omega$. The total realized payoff generated by task $\tau$ is

\begin{equation}
\Pi^\omega_\tau = p_\tau -
\sum_{i\in \mathcal{S}_z} c^{i}_\tau(\boldsymbol{\theta}^\omega),
\end{equation}

where $\mathcal{S}_z \subseteq \mathcal{N}_z$ denotes the set of players
involved in the execution of $\tau$ at slot $z$.

\paragraph{Outer game}
The NO behaves as an ex-ante veto player: without its
infrastructure and orchestration services, no coalition can form or
generate value. The set of selected vehicular executors acts, in
turn, as a \emph{collective veto player}, since without at least one
effective executor the task cannot be processed and no value is created.
In unanimity games, where a coalition consists of $k$ veto players, the
Shapley value assigns an equal share $1/k$ of the total value to each
veto player and zero to all others~\cite{Beal2014}. With two veto blocks this gives the NO half of the total payoff,
\begin{equation}
\label{eq:pi_NO}
\pi^\omega_{\mathrm{NO}}(\tau) = \tfrac{1}{2}\,\Pi^\omega_\tau ,
\end{equation}
a rule symmetric in the sign of $\Pi^\omega_\tau$: the operator absorbs
half of any realized deficit just as the executor group does.

\paragraph{Inner game}
The remaining half is distributed among the selected vehicular executors
in proportion to the cost each of them actually incurred,
\begin{equation}
\label{eq:pi_i_cost}
\pi^\omega_i(\tau) = \tfrac{1}{2}\,\Pi^\omega_\tau \cdot w_i^\omega,
\quad
w_i^\omega := \frac{c^i_\tau(\boldsymbol{\theta}^\omega)}{\sum_{j\in \mathcal{S}_z\setminus\{\mathrm{NO}\}} c^j_\tau(\boldsymbol{\theta}^\omega)},
\end{equation}
for $i \in \mathcal{S}_z\setminus\{\mathrm{NO}\}$. The weight is the
realized energy cost of~\eqref{eq: utility_last}, an observable quantity
recomputed from the post-execution beacons. It is non-negative for every
$\omega$ and sums to one.

\paragraph{Complexity}
The outer game is solved in $O(1)$ and the inner game in
$O(|\mathcal{S}_z|)$. We then verify
stability by checking the core constraints given by the rationality conditions of coalitional games. If one or more constraints
are violated, we compute a corrected imputation by solving a small linear
program that enforces the core conditions, equivalently a projection onto
the feasible set. With $n=|\mathcal{S}_z|$ players and $m$ enforced
constraints, this runs in $\mathrm{poly}(n,m)$ time, e.g.\
$O((n+m)^3)$ with interior-point methods, which for coalitions in the
order of tens of nodes remains compatible with real-time operation.

Payment is collected from the end-users \emph{only} if the task completes
before its deadline. Let \( T^{i,\omega}_\tau \) denote the actual
completion time of task~\(\tau\) by executor~\(i\) under
realization~\(\omega\), and define the indicator
\( \mathbf{1}_{i\tau}^\omega :=
\mathbf{1}_{\{T^{i,\omega}_\tau \le D_\tau\}} \).
This indicator equals~1 only when the task meets its deadline at
sample~$\omega$; otherwise it is~0, and both the executor and the network
operator receive no revenue for that task, regardless of its nominal
allocation. A vehicle misreporting nominal values to be selected in
Stage~1 can only influence its allocation, not its term in the value
function: if the realized $T^{i,\omega}_\tau$ exceeds $D_\tau$ the revenue
contribution vanishes while the realized cost still subtracts, so
sustained misreporting strictly decreases the misreporter's payoff.

For any coalition \( \mathcal{S} \subseteq \mathcal{N}_z \), we define
the set of tasks effectively offloaded within the coalition as:
\begin{equation}
\label{eq:allocated_task_set}
\mathcal{T}_z^{\mathcal{S}} :=
\left\{ \tau \in \mathcal{T}_z \;\middle|\;
\exists\, i \in \mathcal{S} \setminus \{\mathrm{NO}\} \text{ such that }
y^{i,*}_{\tau,z} = 1 \right\}.
\end{equation}

When the NO acts alone, i.e., \(\mathcal{S} = \{\mathrm{NO}\}\), it
incurs only the cost associated with its own allocation and receives a
fraction \(\beta^{\mathrm{NO}}\) of the total payment
\(p_\tau = p(D_\tau, W_\tau)\):

\begin{equation}
\label{eq:value_function_NO_alone}
v_z^\omega(\mathcal{S}) =
\sum_{\tau \in \mathcal{T}^{\mathcal{N}_z}_z}
\left(
\beta^{\mathrm{NO}}\, p_\tau\, \mathbf{1}_{\mathrm{NO},\tau}^\omega
-
c^{\mathrm{NO}}_\tau(\boldsymbol{\theta}^\omega)
\right)
\end{equation}

For any coalition
$\mathcal{S}\subseteq\mathcal{N}_z\setminus \{\mathrm{NO}\}$, each
executor receives its payment share, and the NO's energy cost is not
accounted for in the value function. While
\eqref{eq:value_function_NO_alone} includes all tasks, as they are fully
handled by the NO, here the set of processed tasks depends on the
coalition \(\mathcal{S}\):

\begin{equation}
\label{eq:value_function_NO_not_in}
v_z^\omega(\mathcal{S}) =
\sum_{i \in \mathcal{S}} \sum_{\tau \in \mathcal{T}_z^{\{i\}}}
\left(
\beta^i\, p_\tau\, \mathbf{1}_{i\tau}^\omega
-
c^i_\tau(\boldsymbol{\theta}^\omega)
\right)
\end{equation}

For every coalition \( \mathcal{S} \ni \mathrm{NO} \), the value
function includes the contributions of both the network operator and the
participating executors:

\begin{equation}
\label{eq:value_function_NO_in}
v_z^\omega(\mathcal{S}) =
v_z^\omega(\{\mathrm{NO}\}) + v_z^\omega(\mathcal{S} \setminus \{\mathrm{NO}\})
\end{equation}

We define a \emph{robust cooperative game} as the tuple
\( (\mathcal{N}_z, \mathcal{V}_z) \), where \( \mathcal{N}_z \) is the
set of players and
\( \mathcal{V}_z = \left\{ v_z^\omega \mid \omega \in \Omega \right\} \)
is the set of all possible realizations of the coalition value function
under uncertainty. To guarantee core stability, we prove that both
\emph{efficiency} and \emph{coalitional rationality} hold for every
\( v_z^\omega \in \mathcal{V}_z \). It then follows, by the
Bondareva--Shapley theorem (Proposition 262.1
in~\cite{Osborne1994}), that the core of each game
\( (\mathcal{N}_z, v_z^\omega) \) is non-empty.

\begin{theorem}[Robust Core Non-Emptiness]
\label{thm:core_nonempty_BS_beta}
Fix a slot \(z\in Z\), a realization \(\omega\in\Omega\), and a weight
vector \(\boldsymbol{\beta}=(\beta^i)_{i\in\mathcal N_z}\) with
\(\beta^i\ge0\) and \(\sum_{i\in\mathcal N_z}\beta^i=1\).
Let \(\mathcal C := 2^{\mathcal N_z}\) be the set of all coalitions, and
let \(G_{z,\boldsymbol\beta}^{\omega} = (\mathcal N_z, v_z^\omega)\) be
the TU-game whose value function is defined by
Eqs.~\eqref{eq:value_function_NO_alone}--\eqref{eq:value_function_NO_in},
with task sets \(\mathcal{T}_z^{\mathcal{S}}\) as in
Eq.~\eqref{eq:allocated_task_set}. Assume that each task is assigned to
at most one executor, i.e.,
\(\mathcal{T}_z^{\{i\}} \cap \mathcal{T}_z^{\{j\}} = \varnothing\) for
\(i \neq j\). Then the game \(G_{z,\boldsymbol\beta}^{\omega}\) is
balanced and its core is non-empty.
\end{theorem}

\begin{proof}
Set \(d_{i\tau}:=\beta^i p_\tau\mathbf1^\omega_{i\tau}\),
\(d^\mathrm{NO}_\tau:=\beta^{\mathrm{NO}}p_\tau\mathbf1^\omega_{\mathrm{NO},\tau}\),
\(c_{i\tau}:=c^i_\tau(\boldsymbol\theta^\omega)\),
\(c^\mathrm{NO}_\tau:=c^{\mathrm{NO}}_\tau(\boldsymbol\theta^\omega)\).

Let \(\{\gamma_{\mathcal S}\}_{\mathcal S\in\mathcal C}\) be any balanced
collection: \(\gamma_{\mathcal S}\ge0\) and
\(\sum_{\mathcal S\ni i}\gamma_{\mathcal S}=1\) for all
\(i\in\mathcal N_z\). Using the definition of \(v_z^\omega(\mathcal S)\)
and exchanging sums:

\begin{equation*}
\sum_{\mathcal{S} \in \mathcal{C}} \gamma_{\mathcal{S}}\, v_z^\omega(\mathcal{S})
= \sum_{\mathcal{S} \in \mathcal{C}} \gamma_{\mathcal{S}}
  \left[
      \sum_{i \in \mathcal{S} \setminus \{\mathrm{NO}\}}
      \sum_{\tau \in \mathcal{T}_z^{\{i\}}}
      \left(d_{i\tau} - c_{i\tau}\right)
  \right]
\end{equation*}
\begin{equation*}
\quad + \sum_{\mathcal{S} \in \mathcal{C}} \gamma_{\mathcal{S}}
  \left[
      \mathbf{1}_{\{\mathrm{NO} \in \mathcal{S}\}}
      \sum_{\tau \in \mathcal{T}_z^{\mathcal{N}_z}}
      \left(d^\mathrm{NO}_\tau - c^\mathrm{NO}_\tau\right)
  \right]
\end{equation*}
\begin{equation*}
= \sum_{i \ne \mathrm{NO}}
    \sum_{\tau \in \mathcal{T}_z^{\{i\}}}
    \left(d_{i\tau} - c_{i\tau}\right)
    \sum_{\mathcal{S} \ni i} \gamma_{\mathcal{S}}
\end{equation*}
\begin{equation*}
+ \sum_{\tau \in \mathcal{T}_z^{\mathcal{N}_z}}
    \left(d^\mathrm{NO}_\tau - c^\mathrm{NO}_\tau\right)
    \sum_{\mathcal{S} \ni \mathrm{NO}} \gamma_{\mathcal{S}}
\end{equation*}
\begin{equation*}
= \sum_{i \ne \mathrm{NO}}
    \sum_{\tau \in \mathcal{T}_z^{\{i\}}}
    \left(d_{i\tau} - c_{i\tau}\right)
+ \sum_{\tau \in \mathcal{T}_z^{\mathcal{N}_z}}
    \left(d^\mathrm{NO}_\tau - c^\mathrm{NO}_\tau\right)
\end{equation*}

Define
\begin{align*}
\Sigma_1 &= \sum_{\tau \in \mathcal{T}_z^{\mathcal{N}_z}}
\Bigl( d^\mathrm{NO}_\tau + \sum_{i \ne \mathrm{NO}} d_{i\tau}\Bigr),\\
\Sigma_2 &= \sum_{\tau \in \mathcal{T}_z^{\mathcal{N}_z}}
\Bigl( c^\mathrm{NO}_\tau + \sum_{i \ne \mathrm{NO}} c_{i\tau} \Bigr).
\end{align*}
Then
\[
v_z^\omega(\mathcal{N}_z) = \Sigma_1 - \Sigma_2, \qquad
\sum_{\mathcal{S} \in \mathcal{C}} \gamma_{\mathcal{S}}\, v_z^\omega(\mathcal{S})
= \Sigma_1 - \Sigma_2,
\]
which satisfies the Bondareva--Shapley inequality with equality, so the
game is balanced and
\(\mathrm{Core}(G_{z,\boldsymbol{\beta}}^\omega) \ne \emptyset\).
Furthermore, there exists an imputation vector
\(\boldsymbol{x} = (x^i)_{i \in \mathcal{N}_z}\) such that
\[
\sum_{i \in \mathcal{N}_z} x^i = v_z^\omega(\mathcal{N}_z), \qquad
\sum_{i \in \mathcal{S}} x^i \ge v_z^\omega(\mathcal{S})
\quad \forall \mathcal{S} \subsetneq \mathcal{N}_z. \qquad\blacksquare
\]
\end{proof}

\section{Simulation Framework}
\label{sec:simulation_last}

\begin{table}[t!]
  \centering
  \scriptsize
  \setlength{\tabcolsep}{3pt}
  \begin{tabularx}{\linewidth}{|l|X|}
    \hline
    \textbf{Parameter} & \textbf{Value} \\
    \hline
    \multicolumn{2}{|c|}{\textbf{Architecture and Mobility}} \\
    \hline
    Vehicle average speed & $13.1$\,km/h (Rome downtown, SUMO)~\cite{limitedtimes23} \\
    Cloud nodes / Edge baseline & $1$ / single NVIDIA A30, $3.3\times10^{14}$\,OP/s~\cite{nvidia_a30} \\
    Vehicles / end-users & $10$--$200$ / $25$--$200$ (default $100$)~\cite{vehicular_cloud_computing_2025} \\
    Vehicle peak compute / spare & $3\times10^{14}$\,OP/s~\cite{nvidia_orin} / $1\%$--$20\%$ of peak \\
    gNodeB coverage radius & $500$\,m (5G NR n78 urban micro-cell)~\cite{3gpp38901} \\
    \hline
    \multicolumn{2}{|c|}{\textbf{General Settings}} \\
    \hline
    Duration / warm-up / seeds & $30$\,s / $30$\,s / $10$; allocation round and task window $5$\,ms \\
    Task rate & $1$--$10$ tasks/s per user (Poisson) \\
    \hline
    \multicolumn{2}{|c|}{\textbf{Tasks and Queueing}} \\
    \hline
    Workload $W_\tau$ & discrete law on $\{10^{8},10^{9},10^{10},10^{11},10^{12}\}$\,OP, prob. $\{0.10,0.20,0.30,0.25,0.15\}$ \\
    Mean workload, $\mathrm{CV}^{2}$ & $1.78\times10^{11}$\,OP, $3.80$ (heavy-tailed $\Rightarrow$ M/G/1) \\
    Task size / deadlines & $I_\tau=O_\tau=8000$\,bits~\cite{zhang15} / $D_\tau\in\{16,100,500\}$\,ms~\cite{benameur21} \\
    Payment per task & $p_\tau(D_\tau)\in\{2.63,1.43,1.03\}\,\mu\$$ (AWS Lambda 2026) \\
    \hline
    \multicolumn{2}{|c|}{\textbf{Energy}} \\
    \hline
    Energy price & $0.21$\,\$/kWh (U.S.\ residential avg.)~\cite{eia_electricity_2024} \\
    Controller CPU power & $200$\,W \\
    Compute energy $\zeta^{n}$, vehicle / cloud & $4.35\times10^{-13}$ / $4.24\times10^{-13}$\,J/OP~\cite{nvidia_orin} \\
    Internet energy intensity & $2.7\times10^{-6}$\,J/bit~\cite{Aslan2018} \\
    \hline
    \multicolumn{2}{|c|}{\textbf{Robust Admission}} \\
    \hline
    CVaR level $\alpha$ / radius $\epsilon$ & $0.90$ ($\kappa=1.7549$) / $0$ (misreporting isolated) \\
    Penalty $\lambda_W$ & $5.0$ \\
    Misreporting fraction $\psi$ / intensity & $\{0,0.2,0.4,0.6,0.8\}$ / $0.6$ \\
    \hline
    \multicolumn{2}{|c|}{\textbf{Network}} \\
    \hline
    Channel / carrier / bandwidth & TR~38.901 Urban NLOS, Rayleigh, shadowing $4$\,dB; $3.5$\,GHz (n78) / $100$\,MHz \\
    Internet leg: $\delta_{\text{inet}}$ / $d_{\text{cloud}}$ / $R_{\text{inet}}$ & $35$\,ms one-way~\cite{verizon23} / $10$\,km / $100$\,Gb/s \\
    gNodeB / UE Tx power & $23$\,dBm~\cite{remmaps23}, $2$ streams, $5$\,dB beamforming \\
    Rate model / efficiency & Shannon-based SINR approximation / $0.85$ \\
    Signal speed, air / fiber & $c_{\text{air}}\!\approx\!3\times10^{8}$\,m/s / $c_{\text{fib}}=(2/3)c_{\text{air}}$ \\
    \hline
    \end{tabularx}
    \caption{Simulation parameters.}
  \label{tab:parameters_sim_last}
\end{table}

This section details the simulation framework, aiming at full
repeatability; default parameters are in
Table~\ref{tab:parameters_sim_last}. Unless stated otherwise each point
is the mean over $10$ independent seeds with $95\%$ Student-$t$
confidence intervals, whose overlap is the criterion for statistical
significance. The per-task payment $p_\tau$ is anchored to the
Amazon Web Services Lambda 2026 tariff at $10$ tasks/s per user; lacking
public real-world measurements this remains a theoretical approximation,
kept conservative. Vehicles dedicate only a spare fraction of their
compute capacity to offloading, swept over $\{1,3,5,8,10,15,20\}\%$ of
the total compute capacity, as most resources are reserved for higher-priority functions such
as navigation, safety, and sensor processing.

\subsection{Experimental Scenario}
\label{sec:exp_scenario}
The evaluation reproduces a single 5G micro-cell in a larger urban map.
The SUMO trace of downtown Rome is wider than the $500$\,m serving
radius, so vehicles are born outside coverage, cross the cell and leave:
the set of usable executors is never static. A $30$\,s warm-up precedes
each measured $30$\,s run, so population and queues reach stationarity
before statistics are collected. On top of this substrate we vary vehicle
density from $10$ to $200$, the spare-capacity fraction, and a controlled
fraction $\psi$ of vehicles over-stating their spare capacity in the
Stage-1 beacon at fixed intensity $0.6$; the sweep over end-users from
$25$ to $200$ is reported in the supplementary material. Tasks draw
$W_\tau$ from the discrete law of Table~\ref{tab:parameters_sim_last},
spanning four orders of magnitude; this heterogeneity makes the service
time heavy-tailed and is the reason the M/G/1 model of
Sec.~\ref{sec:time_models} is required.

\subsection{Mobility and Networking}
\label{sec:networking}
Mobility traces come from SUMO on real road and traffic data of Rome,
average speed $13.1$\,km/h~\cite{limitedtimes23}. The model includes
traffic lights, intersections, pedestrian crossings and car-following
dynamics, inducing stop-and-go behavior so that speed and position evolve
non-uniformly, producing realistic dwell-time fluctuations; end-user UEs
are static during each offloading round, their displacement being
negligible (under $1$\,cm for a fast walker over a $5$\,ms round). 5G
performance uses the TR~38.901 urban non-line-of-sight model with path
loss, log-normal shadowing ($4$\,dB) and Rayleigh fading; rates follow a
Shannon-based SINR approximation with interference across active
vehicles, while cloud offloading uses fixed-bandwidth Internet links with
core latency. A single serving gNodeB is co-located with the
\emph{Controller}, with a $500$\,m radius consistent with a 5G NR n78
urban micro-cell~\cite{3gpp38901}; the coverage churn described above
follows, while inter-gNB handovers are excluded by construction. They
are not modeled because their typical duration ($80$--$120$\,ms) exceeds
the $5$\,ms decision window and spans multiple allocation rounds. Within
a round vehicle displacement is negligible ($\sim$$1.85$\,cm), and fading
and shadowing are re-sampled across rounds.

\subsection{Compared Strategies}
\label{sec:compared_strategies}
We benchmark the proposed scheme against five alternatives that isolate,
one component at a time, the contribution of admission and of allocation
optimality; all share the same mobility, channel, queueing and energy
models.

\paragraph{\textsf{DRO} strategy} The full scheme, where
Def.~\ref{def:admission} pre-filters the pairs feeding the exact
allocation~\eqref{eq: optimization_1}.

\paragraph{\textsf{no-DRO} strategy} Solves the same exact allocation
without the pre-filter, trusting every declared capacity.

\paragraph{\textsf{Greedy} strategy} Assigns
each task to the highest-ranked feasible executor by nominal completion
time, with neither the integer program nor admission.

\paragraph{\textsf{Vehicles-first} strategy~\cite{vehicular_cloud_computing_2025}}
Prioritizes offloading to vehicles as a cost-effective alternative to
edge computing, the closest external baseline to our setting; with queue
length fixed to $q^i=1$ per vehicle it keeps loading already-committed
vehicles.

\paragraph{\textsf{Edge-only} strategy} Serves every task from a single
NVIDIA A30 provisioned for the peak task load and is the
edge infrastructure reference of Sec.~\ref{sec:emissions_analysis};
it is not capacity-matched, delivering $3.3\times10^{14}$\,OP/s against
${\sim}10^{14}$ for the in-cell fleet at $3\%$, so it serves essentially
every task.

\paragraph{\textsf{Cloud-only} strategy} Serves every task from the
always-present CC node.

\noindent\textsf{DRO} vs.\ \textsf{no-DRO} isolates what admission buys
with optimality held fixed, \textsf{no-DRO} vs.\ \textsf{Greedy} what
optimality buys without admission.

Two economic inputs are stylized: the per-task payment $p_\tau$ and the
conversion rate $\chi^{i}$. Both enter
every utility linearly, so rescaling them rescales every monetary curve
by the same factor and, energy costs being orders of magnitude below
revenues at all densities (Fig.~\ref{fig: energy_last}), preserves the
sign of the net utility.

\section{Simulation Results}
\label{sec: results-last}

\begin{figure*}[t]
\centering
\begin{subfigure}[b]{0.40\textwidth}
  \centering\includegraphics[width=\linewidth]{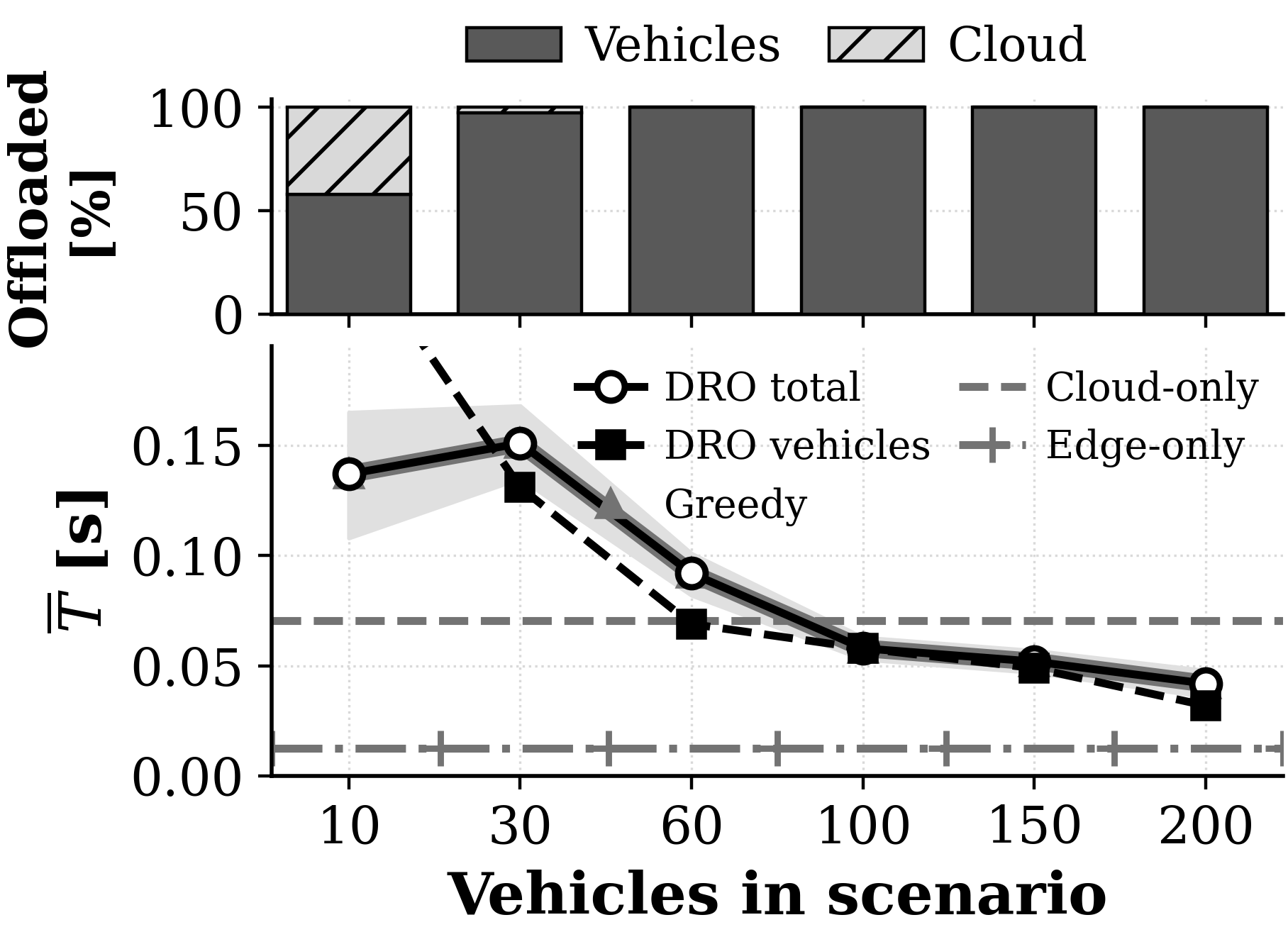}
  \caption{Offloading split and completion time.}
  \label{fig: rate_off_time}
\end{subfigure}\hfill
\begin{subfigure}[b]{0.42\textwidth}
  \centering\includegraphics[width=\linewidth]{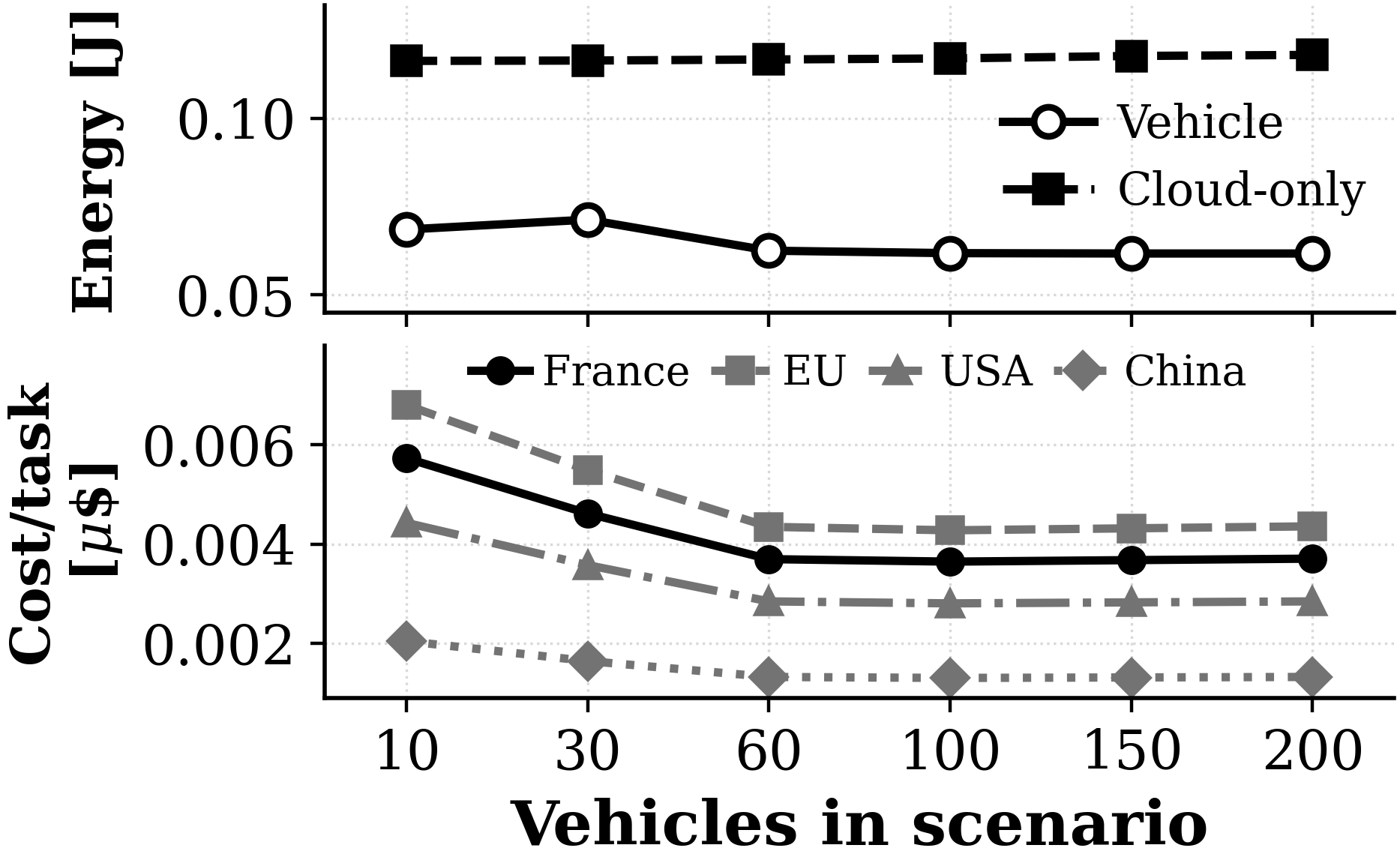}
  \caption{Per-task energy and cost.}
  \label{fig: energy_last}
\end{subfigure}
\caption{VCC evaluation versus vehicle density at $10$ tasks/s per $100$
users and $10\%$ spare capacity. \textbf{(a)} Offloading split (bars) and
mean completion time $\overline{T}$, the average of $T^{n}_\tau$ over the
served tasks and the $10$ seeds, by strategy; the shaded band is the
$95\%$ confidence interval of \textsf{DRO}, \textsf{Greedy} overlaps
\textsf{DRO} at every density, and the cloud floor is its $\sim$$70$\,ms
round trip. \textbf{(b)} Per-task energy on the vehicular and cloud
paths, and its monetary cost under four national electricity prices.}
\label{fig:results_perf}
\end{figure*}

\begin{table}[t]
  \centering\scriptsize\setlength{\tabcolsep}{3.2pt}

  \begin{tabular}{lrrrrrr}
    \toprule
    Vehicles & 10 & 30 & 60 & 100 & 150 & 200 \\
    \midrule
    \multicolumn{7}{l}{\emph{Failure rate} [\%]}\\
    \textsf{DRO} $=$ \textsf{no-DRO} $=$ \textsf{Greedy} & 41 & 48 & 35 & 21 & 16 & 13 \\
    \textsf{Vehicles-first} & 44 & 62 & 46 & 20 & 7 & 5 \\
    \textsf{Cloud-only} & 33 & 33 & 33 & 33 & 33 & 33 \\
    \textsf{Edge-only} & 1.6 & 1.9 & 1.6 & 2.4 & 2.0 & 1.5 \\
    \midrule
    \multicolumn{7}{l}{\emph{Realized utility} [$10^{3}\,\mu\$$ per run]}\\
    Vehicular share & 8.6 & 17.4 & 25.2 & 32.8 & 35.3 & 36.1 \\
    Cloud share & 11.5 & 0.7 & 0.0 & 0.0 & 0.0 & 0.0 \\
    \textbf{Total} & \textbf{20.1} & \textbf{18.1} & \textbf{25.2} & \textbf{32.8} & \textbf{35.3} & \textbf{36.1} \\
    \hspace{1em}$\pm$ half-width & 1.4 & 1.1 & 3.9 & 2.1 & 0.7 & 0.4 \\
    \bottomrule
  \end{tabular}
    \caption{Failure rate [\%] by strategy and aggregate utility split,
  versus vehicle density (honest fleet, $10$ tasks/s per $100$ users,
  $10\%$ spare capacity).
  \textsf{DRO}, \textsf{no-DRO} and \textsf{Greedy} coincide within
  confidence intervals under honest reporting (paired difference
  $-26\pm37$ tasks on $\approx3\times10^{4}$) and are reported once.
  Both blocks are means over the $10$ seeds. The utility is the one
  \emph{realized}, i.e.\ net of the tasks that miss their deadline at
  execution, with the half-width of the $95\%$ Student-$t$ interval on the
  total; the split between the two executors is measured on the single
  seed whose per-assignment trace was retained and applied to the
  realized total.}
  \label{tab:failure_utility}
\end{table}

\subsection{Viability: Offloading Time and Failure Rate}
\label{sec: performance-results-last}

Fig.~\ref{fig: rate_off_time} shows the offloading split and the mean
per-task completion time as a function of vehicle density, at
$10$ tasks/s per $100$ users and $10\%$ spare capacity, a rate that
stresses the system. At $10$ vehicles the few in-coverage executors
saturate: the vehicle-only completion time leaves the plotted range
($250$\,ms), so the allocator diverts close to half of the tasks to the
cloud, which caps the mean at $140$\,ms. As density grows the fleet takes
over, serving $97\%$ of the tasks at $30$ vehicles and essentially all
of them beyond $60$. The mean peaks at $150$\,ms at $30$ vehicles, where
the fleet has taken over but is not yet abundant, then falls
monotonically, crossing the cloud's flat ${\sim}70$\,ms floor between $60$ and
$100$ vehicles and reaching $40$\,ms at $200$; the vehicle-only curve
falls further, to about $30$\,ms. The cloud floor is a pure round trip,
$2\times$ ($35$\,ms core latency plus fiber propagation), cloud service
and queueing being negligible, while \textsf{Edge-only} sits at
$\sim$$10$\,ms throughout, the co-located node paying only the radio
link. \textsf{Greedy} lies on top of \textsf{DRO} at every density, as
Table~\ref{tab:failure_utility} anticipates for an honest fleet, and the
confidence band of \textsf{DRO} is widest at $10$--$30$ vehicles, where the identity of the vehicles in coverage matters most.

Table~\ref{tab:failure_utility} compares the failure rate of the
strategies, a failure being a task not completed within its firm
deadline. The three exact policies coincide within confidence intervals under
honest reporting, falling from $48\%$ at $30$ vehicles to $13\%$ at
$200$; the dip to
$41\%$ at $10$ reflects the cloud absorbing over $40\%$ of the load.
\textsf{Cloud-only} is pinned at $33\%$ at every density: exactly the
$16$\,ms deadline tier of Table~\ref{tab:parameters_sim_last},
unreachable over a ${\sim}70$\,ms round trip. \textsf{Vehicles-first} is
worst at low density ($62\%$ at $30$ vehicles) yet best among VCC at high
density ($5\%$ at $200$): with abundant vehicles, aggressive spreading
beats utility-maximizing concentration on pure QoS, at several times the
wasted cost. \textsf{Edge-only} stays below $3\%$ everywhere. The
residual $13\%$ at $200$ vehicles is the upper bound of the current
architecture under maximum load; a signaling suppression mechanism, in
which vehicles ranked below a density-dependent threshold refrain from
broadcasting offloading-specific beacons, is left to future work.

\subsection{Energy Analysis}
\label{sec: energy-results-last}

Fig.~\ref{fig: energy_last} compares per-task energy as a function of
vehicle density. Vehicular execution costs $60$--$70$\,mJ, with a mild
bump at $30$ vehicles where the fleet has just taken over the load, and
settles near $62$\,mJ from $60$ vehicles on; the cloud path costs
${\sim}118$\,mJ, almost twice as much, and drifts slightly upward with
density as radio contention lengthens the uplink. Metered per
operation~\eqref{eq: energy_node_last} the compute terms are comparable,
and the difference is the Internet leg, whose per-bit
intensity~\cite{Aslan2018} dominates. Keeping execution local therefore
halves per-task energy and relieves core-network congestion, benefiting
an operator that relies on caching or traffic shaping to manage load; a
market analysis by iGR for Mavenir estimates that U.S. mobile operators
could save over \$546 million in five years by offloading up to $32\%$
of video traffic via edge computing instead of the core
network~\cite{igr_mec_savings}. The bottom panel converts energy into
cost through national electricity prices: between about $1$\,n\$ per
task in China and $7$\,n\$ in the EU, with France and the USA in
between, and decreasing with density because at $10$ vehicles close to
half of the tasks pay the cloud path. These are three orders of magnitude
below the payments of Table~\ref{tab:parameters_sim_last}, which is why
energy never drives the economics.

\textbf{Impact on EV battery consumption.} From the vehicle owner's
perspective, a key question is whether offloading meaningfully affects
overall energy consumption. For a Tesla Model~S
($17.5$\,kWh/100\,km~\cite{tesla_power_consumption}) one offloading of
$\approx$$60$\,mJ is $\sim$$10^{-7}\%$ of the energy spent over
$100$\,km. 

\subsection{Profitability: Utility and Incentive Distribution}
\label{sec: revenues-last}

Table~\ref{tab:failure_utility} reports the realized utility per run
versus density, payments net of energy costs and of the tasks that miss
their deadline at execution. The aggregate is \emph{not} monotonic in
density: it falls from $20.1$ to $18.1{\times}10^{3}\,\mu\$$ between $10$
and $30$ vehicles, then grows to $36.1$ at $200$. The minimum sits at the
density at which completion time peaks (Fig.~\ref{fig: rate_off_time})
and the failure rate is highest, $48\%$: at $30$ vehicles the fleet has
taken the load from the cloud without being abundant enough to serve it,
so tasks are admitted, occupy an executor and then miss their deadline,
earning nothing. Below that point the cloud still carries the system: its
share is the majority only at $10$ vehicles, $11.5$ against $8.6$ for the
fleet, and collapses to $0.7$ at $30$ and to zero from $60$ on, as soon
as the fleet is dense enough for the allocator to stop paying the cloud's
$70$\,ms floor. Where vehicles are scarce
the cloud is the profitable executor despite its latency, because it
completes what it accepts; where they are dense the fleet dominates
because it reaches the $16$\,ms tier, which carries the highest payment
and is out of the cloud's reach, receiving $0.6\%$ of its assignments at
$10$ vehicles, $26.1\%$ at $60$ and $29.4\%$ at $200$ against $33\%$
offered. The confidence interval is widest at $60$ vehicles ($\pm3.9$ on
$25.2$), the transition point between the two regimes.

\begin{table}[tbp]
    \centering
    \resizebox{\columnwidth}{!}{%
    \begin{tabular}{l l r r r r r r}
        \toprule
        & & \multicolumn{6}{c}{\textbf{Vehicle density}} \\
        \cmidrule(lr){3-8}
        \textbf{Horizon} & \textbf{Metric} &
        \textbf{10} & \textbf{30} & \textbf{60} &
        \textbf{100} & \textbf{150} & \textbf{200} \\
        \midrule
        \multirow{3}{*}{Monthly (30 d)} &
        Revenue [\$] & 4.13 & 2.78 & 2.02 & 1.57 & 1.13 & 0.87 \\
        & Charging [h] & 3.93 & 2.65 & 1.92 & 1.50 & 1.08 & 0.83 \\
        & Range [km] & 157 & 106 & 77 & 60 & 43 & 33 \\
        \midrule
        \multirow{3}{*}{Yearly (365 d)} &
        Revenue [\$] & 50.2 & 33.9 & 24.5 & 19.1 & 13.8 & 10.6 \\
        & Charging [h] & 47.8 & 32.3 & 23.4 & 18.2 & 13.1 & 10.1 \\
        & Range [km] & 1913 & 1290 & 934 & 729 & 524 & 402 \\
        \midrule
        \multirow{3}{*}{Lifetime (18.4 y)} &
        Revenue [\$] & 924 & 623 & 451 & 352 & 253 & 194 \\
        & Charging [h] & 880 & 594 & 430 & 335 & 241 & 185 \\
        & Range [km] & 35197 & 23741 & 17192 & 13407 & 9636 & 7394 \\
        \bottomrule
    \end{tabular}
    }
    \caption{Net revenues per vehicle, equivalent EV charging hours on a
    7~kW home wallbox, and corresponding driving range for a Tesla
    Model~S, derived from the vehicular utility share of
    Table~\ref{tab:failure_utility} at $160$ runs of $30$\,s per day
    (1h20 offloading/day, vehicle consumption 17.5~kWh/100~km). Recharge
    energy is priced at 0.15~\$/kWh, the retail charging price, distinct
    from the 0.21~\$/kWh of the operator cost model of
    Eq.~\eqref{eq: utility_last}.}
    \label{tab:revenues_and_charging}
\end{table}

\textbf{Incentives to car owners.} A fundamental question in vehicular
task offloading is what would motivate individuals to accept external
tasks on their vehicles. Table~\ref{tab:revenues_and_charging} answers it with three quantities: the revenue a single VO earns on average,
the equivalent recharge time, and the driving range it buys. The
computations assume $1$~h and $20$~min of daily driving in the covered
area~\cite{driving_time} and a vehicle lifetime of $18.4$
years~\cite{Nguyen-Tien2025}. The measured aggregates yield $\$4.13$ to
$\$0.87$ per vehicle-month from $10$ to $200$ vehicles, i.e.\ $157$ to
$33$\,km of monthly range and $\$924$ to $\$194$ ($35\,200$ to
$7\,400$\,km) over the lifetime. Per-vehicle revenues decrease as the
fleet grows, following the law of supply and demand: more vehicles imply
more competition for the same tasks, and the aggregate they share grows
more slowly than their number.

\begin{table}[t]
  \centering\footnotesize\setlength{\tabcolsep}{5pt}
  \begin{tabular}{lrrrr}
    \toprule
    Load [\% of peak] & 100 & 50 & 10 & 1 \\
    \midrule
    Edge A30 (330 TOPS, baseline)  & 0.33 & 0.65 & 3.22 & 32.1 \\
    Edge A100 (624 TOPS) & 0.26 & 0.51 & 2.50 & 24.9 \\
    \textbf{VCC (zero capital)} & \multicolumn{4}{c}{\textbf{0.004 at every load}} \\
    \midrule
    A30 / VCC ratio & $84\times$ & $167\times$ & $826\times$ & $8241\times$ \\
    \bottomrule
  \end{tabular}
    \caption{Cost per served task [$\mu\$$] vs.\ sustained load: the edge
  amortizes its 3-year total cost of ownership over the tasks it serves;
  VCC amortizes nothing.}
  \label{tab:cost_per_task}
\end{table}

\textbf{Cost per task against edge infrastructure.} At full demand
\textsf{Edge-only} attains the highest net utility, serving every task
and collecting every payment even after its amortized cost. The decisive
quantity for an operator is instead the cost per served task
(Table~\ref{tab:cost_per_task}): the edge amortizes a fixed three-year
total cost of ownership over whatever demand materializes, so its
per-task cost is $84\times$ VCC's at full utilization and grows inversely
with load, to three orders of magnitude at the off-peak loads a
peak-provisioned server faces most of the day. VCC's curve is flat,
having no capital to recover, vehicular compute being sunk to the driving
mission: a server pays off only under high, certain, sustained
demand.

\subsection{Robustness to Capacity Misreporting}
\label{sec:robustness-results}

The results so far assume vehicles that report their capacity
truthfully; we now remove this assumption. A fraction $\psi$ of the
vehicles over-states its spare capacity in the Stage-1 beacon by a factor
$0.6$, so as to win assignments it cannot honor. \textbf{Late failures} are tasks admitted in Stage~1 that miss
their deadline at execution: the executor has been consumed and the
end-user learns of the failure only once the deadline has elapsed,
whereas a Stage-1 rejection is immediate and leaves the Controller free
to re-offer the task within its deadline. \textbf{Net utility} is
payments minus operating costs over the run. Table~\ref{tab:divergence}
reports both at two spare-capacity fractions.

\begin{table}[t]
  \centering\scriptsize\setlength{\tabcolsep}{2pt}
  \resizebox{\columnwidth}{!}{%
  \begin{tabular}{llccccc}
    \toprule
    & $\psi$ & 0.0 & 0.2 & 0.4 & 0.6 & 0.8 \\
    \midrule
    \multicolumn{7}{l}{\emph{Late failures} [\% of offered load]}\\
    \multirow{2}{*}{$8\%$}
      & \textsf{DRO}    & 3.7 ($\pm$1.5) & 4.7 ($\pm$1.7) & 4.6 ($\pm$1.8) & \textbf{5.3 ($\pm$1.7)} & 6.0 ($\pm$3.3) \\
      & \textsf{no-DRO} & 3.7 ($\pm$1.5) & 20.4 ($\pm$7.0) & 24.4 ($\pm$6.6) & \textbf{31.1 ($\pm$10.0)} & 24.8 ($\pm$8.8) \\
    \multirow{2}{*}{$3\%$}
      & \textsf{DRO}    & 5.2 ($\pm$1.2) & 4.2 ($\pm$1.2) & 3.2 ($\pm$0.9) & 2.3 ($\pm$1.0) & \textbf{1.5 ($\pm$1.0)} \\
      & \textsf{no-DRO} & 5.2 ($\pm$1.2) & 16.6 ($\pm$5.2) & 23.4 ($\pm$4.9) & 33.3 ($\pm$4.6) & \textbf{39.9 ($\pm$7.1)} \\
    \midrule
    \multicolumn{7}{l}{\emph{Net utility} [$10^{3}\,\mu\$$]}\\
    \multirow{2}{*}{$8\%$}
      & \textsf{DRO}    & 27.6 ($\pm$5.0) & 23.4 ($\pm$1.8) & 22.1 ($\pm$1.4) & 21.1 ($\pm$0.9) & 20.6 ($\pm$1.1) \\
      & \textsf{no-DRO} & 27.6 ($\pm$5.0) & 22.9 ($\pm$4.3) & 21.8 ($\pm$4.3) & 21.7 ($\pm$6.0) & 25.7 ($\pm$4.4) \\
    \multirow{2}{*}{$3\%$}
      & \textsf{DRO}    & 20.7 ($\pm$0.4) & 21.1 ($\pm$0.4) & 21.4 ($\pm$0.3) & 21.7 ($\pm$0.4) & \textbf{22.0 ($\pm$0.3)} \\
      & \textsf{no-DRO} & 20.7 ($\pm$0.4) & 17.0 ($\pm$1.7) & 14.7 ($\pm$1.6) & 11.4 ($\pm$1.6) & \textbf{9.2 ($\pm$2.4)} \\
    \bottomrule
  \end{tabular}}
  \caption{Late failures (\% of the offered load, $\approx 3.0\times10^{4}$
  tasks per run) and net utility [$10^{3}\,\mu\$$] versus the misreporting
  fraction $\psi$, at two spare-capacity fractions; mean ($\pm$~half-width
  of the $95\%$ Student-$t$ CI) over $10$ seeds. \textsf{no-DRO} and
  \textsf{Greedy} coincide exactly throughout, so only \textsf{no-DRO} is
  shown.}
  \label{tab:divergence}
\end{table}

\begin{figure}[t]
  \centering
  \includegraphics[width=\linewidth]{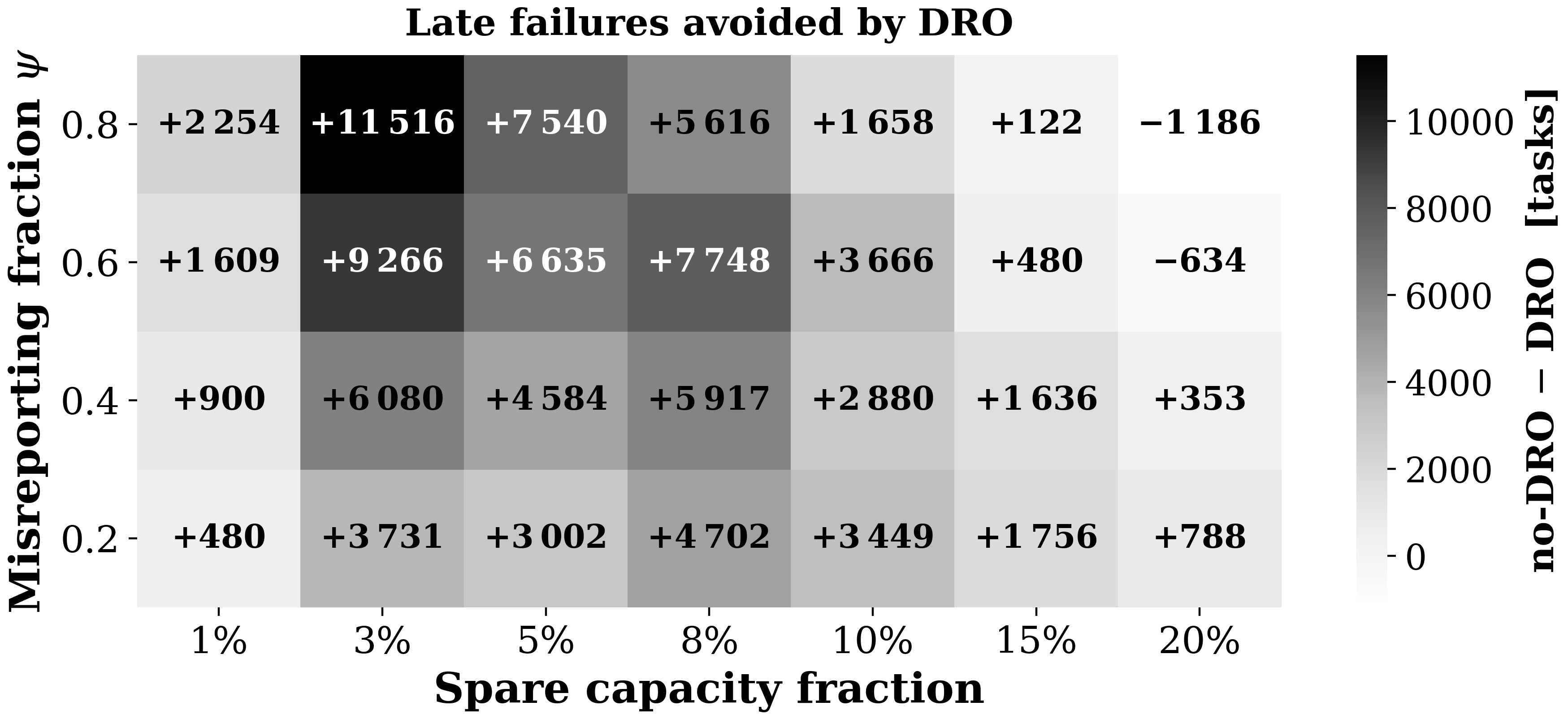}
  \caption{Late failures avoided by \textsf{DRO} (\textsf{no-DRO} minus
  \textsf{DRO}) over the spare capacity and misreporting grid. The honest
  row $\psi=0$ is omitted: with no over-declaration $\lambda_W d^{i}_W$
  vanishes and the two policies solve the same instance, differing only
  by solver tie-breaking (below $0.5\%$ on a single seed).}
  \label{fig:mm1_grid}
\end{figure}

Three facts follow from Table~\ref{tab:divergence}. First, on an honest
fleet the admission rule is inert: at $\psi=0$ the policies coincide
up to solver tie-breaking among equal-utility optima (below $0.5\%$ on a
single seed), because
$\lambda_W d_W^{i}$ vanishes when the observed over-declaration is zero
and admission collapses onto the non-robust optimum. Second, as soon as
some vehicles over-declare the policies separate: at $8\%$ capacity and
$\psi=0.6$, \textsf{DRO} loses $5.3\%$ ($\pm$1.7) of the offered load
against $31.1\%$ ($\pm$10.0), a sixfold reduction, the intervals disjoint
for every $\psi\ge0.2$. Fig.~\ref{fig:mm1_grid} maps the failures
avoided over the whole grid: the protection is largest at $3\%$ spare
capacity, where it grows monotonically with $\psi$ up to $11\,516$ tasks
at $\psi=0.8$, about $38\%$ of the offered load, and it stays in the
thousands across the $3$--$10\%$ band. The $1\%$ column is small for a
different reason: at that capacity almost no task can be served by any
vehicle, so there is little left to protect. From $15\%$ the avoided
failures fade and at $20\%$ the sign inverts, \textsf{DRO} losing
$634$ and $1\,186$ tasks more than \textsf{no-DRO} at $\psi=0.6$ and
$0.8$: in abundance an over-declaring vehicle would have completed the
task anyway, and rejecting it only concentrates the load on the remaining
executors. Third, in every cell of the campaign
\textsf{Greedy} and \textsf{no-DRO} produce identical outcomes, task
counts and utilities coinciding exactly: an optimal allocator that trusts
inflated capacities fails exactly like a greedy one, so the protection
comes from the admission rule and not from the optimality of the
allocation.

The effect peaks in scarcity: at $3\%$ the \textsf{DRO} losses
\emph{decrease} with misreporting, from $5.2\%$ to $1.5\%$ of the load,
while the baselines climb to $39.9\%$. Scarce capacity is where a
misreport is most damaging: with no slack to absorb the miss, an
admitted misreporter saturates its M/G/1 queue and the tasks it holds
are lost, whereas the rule turns away the very vehicles that would fail.
The economics follow the same regime map, with a sign change. At scarce
capacity the \textsf{DRO} net utility rises with misreporting, every
avoided late failure having otherwise consumed resources for no reward,
while the baselines collapse. At abundant capacity an over-declaring
vehicle often completes the task anyway, so rejecting it forfeits
revenue: the gain of \textsf{DRO} over \textsf{no-DRO} at $\psi=0.6$ goes
from $+10.3$ ($\pm$1.6)\,$10^{3}\mu\$$ at $3\%$ to $-0.7$ ($\pm$5.6) at
$8\%$, where the interval still contains zero, and becomes significantly
negative from $15\%$ spare capacity onward. Under pure energy-cost
accounting robust admission therefore pays in the scarce regime, the
realistic one for a fleet reserving most of its compute for driving, and
costs money in abundance. An adaptive penalty relaxing $\lambda_W$ with
measured load and misreporting was designed to recover utility in
abundance; the data do not support it, since relaxing the radius admits
exactly the vehicles that then miss their deadlines, and it is reported
as a tested and rejected variant in the supplementary material, together
with the attribution of every generated task to its fate and a
sensitivity study over adversary intensity, fleet heterogeneity and load,
in all fifteen configurations of which \textsf{DRO} reduces late failures
significantly.

\section{Carbon Emissions Analysis}
\label{sec:emissions_analysis}

\begin{table}[t]
    \centering
    \footnotesize
    \setlength{\tabcolsep}{3.5pt}
    \begin{tabular}{lrrrrrr@{\hskip 8pt}rr}
        \toprule
        & \multicolumn{6}{c}{\textbf{Annual VCC [g], vehicle density}}
        & \multicolumn{2}{c}{\textbf{5-yr (d.\ 60)}}\\
        \cmidrule(lr){2-7}\cmidrule(lr){8-9}
        \textbf{Region} &
        \textbf{10} & \textbf{30} & \textbf{60} &
        \textbf{100} & \textbf{150} & \textbf{200} &
        \textbf{Edge [kg]} & \textbf{Sav.}\\
        \midrule
        France & 10.1 & 8.3 & 6.2 & 5.7 & 5.8 & 5.8 & 1182 & 99.99\% \\
        EU-27  & 125 & 102 & 77 & 71 & 71 & 72 & 4095 & 99.99\% \\
        USA    & 191 & 156 & 117 & 108 & 108 & 109 & 5759 & 99.99\% \\
        China  & 300 & 246 & 185 & 170 & 171 & 172 & 8536 & 99.99\% \\
        \bottomrule
    \end{tabular}
    \caption{Left: annual CO$_2$ per vehicle (g): mean per-task energy
    of Fig.~\ref{fig: energy_last} times $4\,800$\,s/day of offloading,
    $365$ days, $10$~tasks/s and the regional grid intensity. Right:
    five-year footprint of an edge server, computed as
    $925$\,kg embodied plus $13\,100$\,kWh times the regional intensity,
    and VCC saving at density~60 under full attribution of the server to
    offloading.}
    \label{tab:co2_yearly}\label{tab:vcc_savings}
\end{table}

\subsection{Estimated Carbon Emissions per Vehicle}

Carbon emissions per vehicle are estimated by scaling the average energy
used per task by the tasks a vehicle serves at $10$~tasks/s over a daily
offloading duration of 1~h and 20~min, and aggregating the result
annually. Emissions are then computed using
national grid carbon intensity, which expresses CO$_2$ emissions per
consumed kWh. The carbon intensity of electricity production
varies across countries: \textbf{France} emits $19.6$\,g\,CO$_2$/kWh,
primarily due to its energy mix (5\% fossil, 25\% renewables, 70\%
nuclear)~\cite{RTE2025}; the \textbf{EU average} is
$242$\,g\,CO$_2$/kWh (37\% fossil, 39\% renewables, 24\%
nuclear)~\cite{EEA2025}; the \textbf{United States} reaches
$369$\,g\,CO$_2$/kWh (39\% fossil, 39\% renewables, 22\%
nuclear)~\cite{EPAeGRID2025}; and \textbf{China} peaks at
$581$\,g\,CO$_2$/kWh (65\% fossil, 31\% renewables, 4\%
nuclear)~\cite{IEA2024}. As reflected in Table~\ref{tab:co2_yearly}, the
environmental cost of vehicular cloud computing ranges from $6$ to
$10$\,g\,CO$_2$ per vehicle and year in France to $170$--$300$\,g in
China. Emissions decrease with density: at $10$ vehicles close to half of
the tasks travel the cloud path at ${\sim}118$\,mJ, while from $60$
vehicles on the fleet serves them locally at ${\sim}62$\,mJ
(Fig.~\ref{fig: energy_last}), so a denser fleet is also a cleaner one.

\subsection{Vehicular vs.\ Edge Computing: Carbon Emissions}
\label{sec:lca-edge-vcc}

We contrast the greenhouse gas footprint of an Edge server with
that of a vehicle used for computation offloading. Following the
ISO~14040 life-cycle assessment over a five-year functional
unit~\cite{boavizta2021}, we account for production, operation, and
end-of-life phases. The \textbf{Edge server} is a 1U dual-CPU node with
256\,GB of RAM and two SSDs, drawing $300$\,W on average, for
approximately $13\,100$\,kWh over five years; its manufacturing emissions
are estimated between $900$ and $1\,250$\,kg\,CO$_2$ and end-of-life
disposal provides a recycling credit between $50$ and $100$\,kg, giving a
net embodied term of $925$\,kg taken at the midpoint and held constant
across regions. The Edge column of Table~\ref{tab:vcc_savings} is this
embodied term plus the operational energy at the regional intensity. In
contrast, a \textbf{VCC vehicle} reuses existing hardware and operates
for 1h20 per day, with annual emissions of a few hundred grams at most
across all densities. Besides emissions, VCC avoids the material
demand associated with Edge server production: each saved server
eliminates the need for about $12$\,kg of steel, $1.5$\,kg of aluminum,
$1$\,kg of copper, and trace amounts of silicon, tantalum, neodymium,
gold and silver, whose extraction and processing entail severe
environmental and geopolitical costs. Even in carbon-intensive
environments, using VCC instead of deploying edge servers cuts
life-cycle emissions by over $99\%$, up to $8.5$\,t\,CO$_2$ per vehicle
over five years.

\subsection{Peak-Load Attribution}
\label{sec:peak}
The Edge column of Table~\ref{tab:vcc_savings} charges the entire
embodied and operational footprint of a newly provisioned server to
offloading, an upper bound. In practice the same server is not consumed
by offloading alone. Since infrastructure is sized for the peak demand it
must absorb, the attribution consistent with that sizing is pro rata to
the peak: we charge to offloading the fraction $f_{\text{peak}}$ of the
footprint equal to the peak compute share that offloading occupies on the
device, measured on the raw traces of the campaign; the reproducibility
script accompanying the code emits Table~\ref{tab:co2_yearly} for any
measured $f_{\text{peak}}$. This scales the Edge column and the savings
accordingly; because vehicular embodied emissions are zero by
construction and the VCC operational term is three orders of magnitude
smaller, the peak-share comparison lands in the same direction with a
deployment-specific margin. Against a shared multi-access edge computing
server already deployed for other purposes the same reasoning applies:
only the marginal, peak-attributed footprint should be allocated, and
the advantage narrows without inverting, shared servers still needing
power, cooling and end-of-life accounting that scales with workload.

\section{Conclusion}
\label{sec: conclusion-last}

This paper asked whether reusing the spare computing capacity of vehicles
is viable, profitable and sustainable.

\textit{On viability}, vehicles absorb most of the offloaded traffic while keeping
delay within tens of milliseconds, satisfying applications with stringent
latency constraints, and the queueing term makes the load already carried
by an executor visible to the allocation decision.

\textit{On profitability}, the economic incentives more than offset participation
costs even under high competition, yielding between $157$ and $33$\,km
of monthly driving range per vehicle, $\$4.13$ to $\$0.87$ per month,
depending on density, and $\$924$ to $\$194$ ($35\,200$ to
$7\,400$\,km) over the vehicle's lifetime at 1h20 of daily
driving in the covered area. The realized aggregate is not monotonic in
density, being lowest at $30$ vehicles, where the fleet has taken the
load from the cloud without being able to serve it. The revenue-sharing rule keeps the imputation in the core of the
game at every slot, and the admission rule protects the deadline
guarantee against over-stated declarations without altering the outcome
on an honest fleet, cutting deadline misses sixfold once vehicles
over-declare; the protection comes from admission, an optimal allocator
that trusts inflated capacities failing exactly like a greedy one. Beyond
roughly $8\%$ spare capacity the rule costs money under pure energy-cost
accounting, and an adaptive variant designed to recover that loss is not
supported by the data.

\textit{On sustainability}, the added emissions are small: a vehicle emits
between $6$ and $300$\,g\,CO\textsubscript{2} per year depending on
grid carbon intensity and fleet density, while replacing an edge server saves up
to $8.5$\,t\,CO\textsubscript{2} over its life cycle and reduces the
demand for rare elements.

The natural
next steps are a prototype and a field validation, multi-cell
handover-aware allocation, signaling suppression at high density, and
trust mechanisms layered on the ex-post verification the scheme already
performs.

\section*{Acknowledgement}
This work was funded by the French government under the France 2030 ANR program “PEPR Networks of the Future” (ref. ANR-22-PEFT-0003).

\section*{Use of AI-generated content.}
The simulator and the allocation and revenue-sharing code were written by
the authors, who also ran the simulation campaign. A large language model was used during the preparation of
this manuscript to edit and restructure text drafted by the authors, and to
correct and optimise the scripts. Every script was reviewed and executed by
the authors, and every sentence, equation and number in the paper was
verified by them. The schematic in
Fig.~1(a) was drawn by hand (digitally) and its resolution was increased using an
AI-based image upscaling tool; no element of the drawing was AI-generated.

\bibliographystyle{IEEEtran}
\bibliography{refs}

\end{document}